\documentclass[%
 reprint,
 amsmath,amssymb,
]{revtex4-2}
\usepackage{mathbbol}
\usepackage{graphicx}
\usepackage{dcolumn}
\usepackage{bm}
\usepackage{hyperref}
\usepackage[mathlines]{lineno}
\hypersetup{
	colorlinks=true,
	linkcolor=blue,
	filecolor=blue,      
	urlcolor=blue,
	pdftitle={Overleaf Example},
	pdfpagemode=FullScreen,
}
\usepackage{mathtools}
\usepackage{amsthm}
\usepackage{tikz}
\usepackage{tikz-cd}
\usepackage{tikzit}
\tikzstyle{every picture}=[baseline=-0.25em,]

\tikzstyle{box}=[shape=rectangle, text height=1.5ex, text depth=0.25ex, yshift=0.5mm, fill=white, draw=black, minimum height=5mm, yshift=-0.5mm, minimum width=5mm, font={\small}]
\tikzstyle{gate}=[shape=rectangle, text height=1.5ex, text depth=0.25ex, yshift=0.5mm, fill=white, draw=black, minimum height=5mm, yshift=-0.5mm, minimum width=5mm, font={\small}, tikzit category=circuit]
\tikzstyle{big gate}=[shape=rectangle, text height=1.5ex, text depth=0.25ex, yshift=0.5mm, fill=white, draw=black, minimum height=10mm, yshift=-0.5mm, minimum width=5mm, font={\small}, tikzit category=circuit]
\tikzstyle{Z dot}=[inner sep=0mm, minimum size=2mm, shape=circle, draw=black, fill={rgb,255: red,221; green,255; blue,221}, tikzit category=zx]
\tikzstyle{Z phase dot}=[minimum size=5mm, font={\footnotesize\boldmath}, shape=rectangle, rounded corners=2mm, inner sep=0.2mm, outer sep=-2mm, scale=0.8, tikzit shape=circle, draw=black, fill={rgb,255: red,221; green,255; blue,221}, tikzit draw=blue, tikzit category=zx]
\tikzstyle{X dot}=[Z dot, shape=circle, draw=black, fill={rgb,255: red,255; green,136; blue,136}, tikzit category=zx]
\tikzstyle{X phase dot}=[Z phase dot, tikzit shape=circle, tikzit draw=blue, fill={rgb,255: red,255; green,136; blue,136}, font={\footnotesize\boldmath}, tikzit category=zx]
\tikzstyle{hadamard}=[fill=yellow, draw=black, shape=rectangle, inner sep=0.6mm, minimum height=1.5mm, minimum width=1.5mm, tikzit category=zx]
\tikzstyle{paulibox}=[fill={rgb,255: red,221; green,221; blue,255}, draw=black, shape=rectangle, inner sep=0.6mm, minimum height=5mm, minimum width=5mm, font={\footnotesize}, text height=1.5ex, text depth=0.25ex, tikzit category=zx]
\tikzstyle{vertex}=[inner sep=0mm, minimum size=1mm, shape=circle, draw=black, fill=black, tikzit category=misc]
\tikzstyle{vertex set}=[inner sep=0mm, minimum size=1mm, shape=circle, draw=black, fill=white, font={\footnotesize\boldmath}, tikzit category=misc]
\tikzstyle{small black dot}=[fill=black, draw=black, shape=circle, inner sep=0pt, minimum width=1.2mm, tikzit category=circuit]
\tikzstyle{cnot ctrl}=[fill=black, draw=black, shape=circle, inner sep=0pt, minimum width=1.2mm, tikzit category=circuit]
\tikzstyle{cnot targ}=[fill=white, draw=white, shape=circle, tikzit category=circuit, label={center:$\oplus$}, inner sep=0pt, minimum width=2.1mm, tikzit fill={rgb,255: red,102; green,204; blue,255}, tikzit draw=black]
\tikzstyle{ket}=[fill=white, draw=black, shape=regular polygon, regular polygon sides=3, regular polygon rotate=-30, scale=0.7, inner sep=1pt, tikzit category=circuit, tikzit shape=rectangle, tikzit fill=green]
\tikzstyle{bra}=[fill=white, draw=black, shape=regular polygon, regular polygon sides=3, regular polygon rotate=30, scale=0.7, inner sep=1pt, tikzit category=circuit, tikzit shape=rectangle, tikzit fill=red]
\tikzstyle{scalar}=[shape=rectangle, text height=1.5ex, text depth=0.25ex, yshift=0.5mm, fill=white, draw=black, minimum height=5mm, yshift=-0.5mm, minimum width=5mm, font={\small}]
\tikzstyle{clabel}=[fill=white, draw=none, shape=rectangle, tikzit fill={rgb,255: red,56; green,255; blue,242}, font={\footnotesize}, inner sep=1pt, tikzit category=labels]
\tikzstyle{empty diagram}=[draw={gray!40!white}, dashed, shape=rectangle, minimum width=1cm, minimum height=1cm, tikzit category=misc]

\tikzstyle{simple}=[-]
\tikzstyle{hadamard edge}=[-, dashed, dash pattern=on 2pt off 0.5pt, thick, draw={rgb,255: red,68; green,136; blue,255}]
\tikzstyle{box edge}=[-, dashed, dash pattern=on 2pt off 0.5pt, thick, draw={rgb,255: red,203; green,192; blue,225}]
\tikzstyle{brace edge}=[-, tikzit draw=blue, decorate, decoration={brace,amplitude=1mm,raise=-1mm}]
\tikzstyle{diredge}=[->]
\tikzstyle{double edge}=[-, double, shorten <=-1mm, shorten >=-1mm, double distance=2pt]
\tikzstyle{gray edge}=[-, {gray!60!white}]
\tikzstyle{pointer edge}=[->, very thick, gray]
\tikzstyle{boldedge}=[-, line width=1.6pt, shorten <=-0.17mm, shorten >=-0.17mm]

\input{quantum.tikzdefs}
\usepackage[all,2cell,arrow,matrix]{xy} 
\usepackage[small,nohug,heads=vee]{diagrams} 
\diagramstyle[labelstyle=\scriptstyle]
\usetikzlibrary{braids}
\newtheorem{definition}{Definition}
\newtheorem{theorem}{Theorem}
\newtheorem{lemma}{Lemma}
\newtheorem{rem}{Remark}

\newtheorem{corollary}{Corollary}

\newcommand{\cat}{\mathcal{C}}
\newcommand{\vc}{\textbf{Vec}}

\newcommand{\hilb}{\textbf{Hilb}}

\begin{document}


\title{The ZX-calculus as a Language for Topological Quantum Computation}
\thanks{The authors thank the reviewers for their feedback. AK would like to acknowledge support from AFOSR grant FA2386-18-1-4028. }

\author{Fatimah Rita Ahmadi}
\email{f.ahmadi@imperial.ac.uk}
\author{Aleks Kissinger}
\affiliation{
Department of Computer Science, University of Oxford}
\begin{abstract}
Unitary fusion categories formalise the algebraic theory of topological quantum computation. These categories come naturally enriched in a subcategory of the category of Hilbert spaces, and by looking at this subcategory, one can identify a collection of generators for implementing quantum computation. We represent such generators for the Fibonacci and Ising models, namely the encoding of qubits and the associated braid group representations, with the ZX-calculus and show that in both cases, the Yang-Baxter equation is directly connected to an important rule in the complete ZX-calculus known as the P-rule, which enables one to interchange the phase gates defined with respect to complementary bases. In the Ising case, this reduces to a familiar rule relating two distinct Euler decompositions of the Hadamard gate as $\pi/2$ Z- and X-phase gates, whereas in the Fibonacci case, we give a previously unconsidered exact solution of the P-rule involving the Golden ratio. We demonstrate the utility of these representations by giving graphical derivations of the single-qubit braid equations for Fibonacci anyons and the single- and two-qubit braid equations for Ising anyons. We furthermore present a fully graphical procedure for simulating and simplifying braids with the ZX-representation of Fibonacci anyons.  
\end{abstract}

\keywords{topological quantum computation, fusion category, ZX-calculus, anyon, braiding}
\maketitle


\section{\label{sec:level1}Introduction}

One of the biggest obstacles in constructing a quantum computer is the prevalence of errors due to the delicate and highly-sensitive nature of the physical systems typically encoded quantum data. There are essentially two ways to cope with errors in quantum computation. The first is to develop hardware that produces as few errors as possible and deal with the remaining errors via quantum error correction, i.e. by introducing some extra quantum memory and computational overhead to correct errors during the computation. The second, and in some ways related approach, is to develop models of computation that are inherently robust to local errors. Topological quantum computing~\cite{sarma2006} aims to be such a model.

The main idea behind topological quantum computation is to exploit features of the excitations of 2-dimensional topologically ordered systems called anyons. Qubits are encoded in the mutual states of anyons, and computation is done by exchanging, or \textit{braiding} them over time~\cite{rowell2016}. Since representations of braid groups formalise anyonic statistics and topological charges of mutual states of anyons are robust under local perturbations, the computation is intrinsically resistant to errors.

Anyons break the standard physicists' intuition about bosonic and fermionic statistics. Exchanging a pair of identical bosons leaves the quantum state invariant, whereas exchanging a pair of fermions introduces a global phase of $-1$. However, exchanging abelian anyons multiplies the state with a non-trivial phase factor $e^{i\theta}$, and exchanging non-abelian anyons applies a more generic unitary operation. In other words, we transition from the symmetric group to the braid group, which carries statistical properties of anyons. 

There is a significant difference between computation with anyons and other particles, for example, photons. In computing with other particles,  we encode qubits in the states of particles. For example, in the up/down spins of electrons or right/left photon polarization.  In topological quantum computation, qubits are encoded in the mutual statistics of anyons. Hence, in this sense, one does not need to know the Hilbert space of or an exact state of anyons. Qubits are processes from an initial configuration of anyons to a final configuration. In other words, the vector space of processes between different configurations is the computational space. 

The language of category theory is well suited to formalise anyons and processes between them since categories generally are agnostic to object properties~\cite{mac2013}. In particular, unitary fusion categories offer an elegant formulation for anyonic models, and many properties of anyons can be studied purely in terms of the structure of such categories.

Another project started by Abramsky and Coecke tries to formalise quantum mechanics with category theory~\cite{abramsky2004}. Categorical quantum mechanics attempts to shift the idea from states to processes. It proposes that quantum properties can be captured by examining the structural and algebraic properties of the category of Hilbert spaces, \hilb. A notable feature of the CQM programme is the focus on graphical calculi, which use the string diagram notation native to monoidal categories to capture and reason about processes. Perhaps the most notable graphical calculus is the ZX-calculus~\cite{coecke2011interacting}, which has been used extensively in the study of quantum computation. We will discuss this further in the following sections. 

A natural question to ask is the connection between these two pictures. How is TQC related to CQM? Or one can ask, given that TQC is a particular model of quantum computation, and CQM is the general categorical formalism for any model of quantum computation, is there any functor between these two categories? If yes, what are the properties of this functor? 

One may immediately declare a positive answer, and one may correctly suspect the functor is an inclusion. But this is not immediate since objects of the TQC category, as it has appeared in the literature, are anyon types. However, as mentioned earlier, the spaces assigned to anyon types are not generally finite-dimensional nor computational space. Hence, anyons cannot be mapped through a functor to finite dimensional vector spaces, i.e. objects of \textbf{Hilb}. Moreover, any map which sends hom-spaces of TQC to objects of \textbf{Hilb} is not functorial. 

Here, we attempt to clarify this subtle point in the literature. 
We make a clear distinction between an anyonic category and define the category of topological quantum computation, and demonstrate that the topological quantum computation category with our definition is a subcategory of \textbf{Hilb}. In other words, we distinguish between an anyonic category as the hardware of computation and topological quantum computing as a model of computation. Subsequently, we can clearly observe any anyonic category can have multiple topological quantum computation categories defined by fusion spaces, fusion and braiding matrices. 

Having this sketch in mind, we represent the elements of two well-defined models of topological quantum computation, namely Fibonacci and Ising, with the ZX-calculus. We show the ZX-equivalent of Fibonacci and Ising single qubit gates. We also derive a new P-rule for Fibonacci anyons that precisely reduces and simplifies a chain of Fibonacci angles graphically without explicitly working with irrational angles such as the golden ratio. We moreover graphically prove the Yang-Baxter equation for both cases using the ZX-calculus rewrite rules and corresponding Fibonacci and Ising P-rules. This leads us to the idea of using the ZX-representation of anyons for simplifying braids, which is useful while simulating braids with a quantum computer. 

\subsection{\label{sec:level2}Topological Quantum Computation}

We assume the reader has an elementary knowledge of category theory at the level e.g. of Leinster's book \textit{Basic Category Theory}~\cite{leinster2014}.  

In this section, We first provide some necessary background on monoidal, semi-additive, semi-simple, braided, ribbon, and unitary categories. We describe the topological quantum computing category as it has appeared in the literature. We then extract and define a category based on this category solely for quantum computing purposes. We clearly differentiate between these two categories by calling the former category an anyonic theory category and the latter a TQC category. 

\begin{definition}
A \textbf{monoidal category}  $(\mathcal{C}, \otimes, F, l, r, I)$ is a category with the unit object, $I$, tensor functor $\otimes: \mathcal{C} \times \mathcal{C} \longrightarrow \mathcal{C}$ and natural transformations $F: (-\otimes-) \otimes- \longrightarrow -\otimes (- \otimes-)$ as associator and left- and right-unitors $l: I \otimes - \longrightarrow -$  and $r: - \otimes I \longrightarrow -$ such that they satisfy pentagonal equation Figure~\ref{fig:penta} and triangle equation Figure~\ref{fig:triangle}. If $F$, $r$, and $l$ are identity morphisms, the category is \textbf{strict}. 
\begin{figure}[t] 
\begin{center}\scalebox{0.8}{
$\begin{array}{c}
\xymatrix{
& ((X \otimes Y) \otimes Z) \otimes W \ar[dl]_{F_{X, Y, Z} \otimes \id_W} \ar[dr]^{F_{X \otimes Y,Z,W}} \\
(X \otimes (Y \otimes Z)) \otimes W \ar[d]_{F_{X,Y \otimes Z,W}} & & (X \otimes Y) \otimes (Z \otimes W) \ar[d]^{F_{X,Y,Z \otimes W}} \\
X \otimes ((Y \otimes Z) \otimes W) \ar[rr]_{\id_X \otimes F_{Y,Z,W}} & & X \otimes (Y \otimes (Z \otimes W))}
\end{array}$}
\end{center}
\caption{Pentagonal equation.}\label{fig:penta}
\end{figure}
\begin{figure}[!t]
\begin{center}\scalebox{0.8}{
$
\begin{array}{c}
\xymatrix{
(X \otimes I) \otimes Y \ar[rr]^{F_{X,I,Y}} \ar[dr]_{r_X \otimes \id_Y} & & X \otimes (I \otimes Y) \ar[dl]^{\id_X \otimes l_Y} \\
& X \otimes Y}
\end{array}$
}
\end{center}
\caption{Triangle equation.} \label{fig:triangle}
\end{figure}
	
\end{definition}

A well-developed example of a monoidal category is the category of $\mathbb{k}$-vector spaces over a field $\mathbb{k}$. There are at least two ways to make this into a monoidal category: setting the monoidal product to be the tensor product (in which case $I$ becomes the 1-dimensional space $\mathbb{k}$) or setting it to be the direct sum (in which case $I$ is the 0-dimensional space $\textbf{0}$).

\begin{definition}
Let~$\cat$ be a monoidal category, and $X$ be an object in $\cat$. A right dual to $X$ is an object $X^*$ with two morphisms
	\begin{align}
		&	e_X: X^* \otimes X \longrightarrow I \\
		&	i_X: I \longrightarrow X \otimes X^*
	\end{align}
	such that $(id_X \otimes e_X)\circ(i_X \otimes id_{X})=id_X$ and $(e_X \otimes id_{X^*}) \circ (id_{X^*} \otimes i_X) = id_{X^*}$. 
\end{definition}

A left dual can be defined correspondingly by interchanging the roles of $X$ and $X^*$ in the definition above.

\begin{definition}
A monoidal category is \textbf{rigid} if every object has the equivalent right and left duals.  
\end{definition}
Note that $*$ is just an equivalence from $\cat$ to $\mathcal{C}^{opp}$. One can, furthermore, prove dual objects are unique up to a unique isomorphism. We defined a more restrictive version of rigidity; in general, left and right duals can be different~\cite{bakalov}. 

Next, we need to define a notion of addition between objects. This category can, furthermore, demand an addition between morphisms.
\begin{definition}
A \textbf{semi-additive category} is a category whose objects have  direct sums, i.e.  $X \oplus Y$.
\end{definition}

A direct sum, which is also sometimes called a biproduct, is an object that is both a product and coproduct of two objects in a coherent way. As soon as one has biproducts, it is possible to take sums of morphisms. We can take this further and assume that one can form arbitrary linear combinations of morphisms and that composition of morphisms $g \circ f$ is linear in both $f$ and $g$. In this case, we say a category is \textit{enriched in Vector spaces}. Adding this to the notion above, we obtain the following definition. 

\begin{definition}An \textbf{additive category} is a semi-additive category whose objects have finite direct sums, and it is enriched over the category of vector spaces.
\end{definition}
\begin{definition}
An \textbf{abelian category} is an additive category if every morphism has a kernel and a cokernel, and every monomorphism is a kernel and every epimorphism is a cokernel. 
\end{definition}

Note that an alternative definition of abelian categories one finds in the literature assumes that categories are enriched just over abelian groups. Here, we adopt the convention e.g. of~\cite{bakalov} and assume full vector space enrichment. In the next part, we define semi-simple categories. The semi-simplicity ensures objects in the desirable category are restricted to only anyon types. 
\begin{definition}
	A \textbf{sub-object} of an object $X$ is an isomorphism class of monomorphisms.
	$$
	i: Y \hookrightarrow X
	$$ 
\end{definition}
For example, in the category of vector spaces, vector space, $\textbf{0}$, is a subobject of all objects, as there is a unique map from zero vector space to all vector spaces. In the category of sets, given a set, $A$, its subsets are subobjects of $A$.
\begin{definition}
A \textbf{zero object} is an object such that for every object in the category, there exist unique maps from/to the zero object to/from that object.
\end{definition}
\begin{rem}
Note that the zero object is the zero object of the direct sum, meaning $A\oplus 0 \cong A$. 
\end{rem}
An example is the zero vector space in the category of vector spaces. The single-element set is in the category of pointed sets (a pointed set is a set with a chosen element in the set). 
\begin{definition}
A \textbf{simple object} is an object whose sub-objects are only the zero object and itself. Let us denote simple objects with some indices, $X_i$.  
\end{definition}
\begin{definition}
An abelian category is \textbf{semi-simple} if any object $X$ is isomorphic to a direct sum of simple objects. 
\begin{equation}
	X \cong \bigoplus_{i \in I} N_i X_i 
\end{equation}
\end{definition} 
The non-negative integer numbers coefficients behind $X_i$ count the number of nonequivalent projections and injections in hom-sets, i.e. $N_i = dim(hom(X, X_i))$.  In other words, it counts the number of unique isomorphisms between $X$ to $X_i$. 

For example, if $X \cong 2 X_i \oplus 4 X_j$, it means there exist 2 unique morphisms between $X$ and $X_i$ and 4 unique morphisms between $X$ and $X_j$. (Here we restrict ourselves to multiplicity-free models.) From the definition of simple objects, we conclude the only morphism between two simple objects is the zero morphism, so-called \textbf{Schur's Lemma} \cite{bakalov}. 
\begin{equation}
hom(X_i, X_j) \cong \mathbf{\delta_{ij} \mathbb{C}}, 
\end{equation}   
where $\delta_{ij}=0$ if $i\neq j$, and otherwise, $\delta_{ii}=1$. The definitions laid out in the previous parts provide ample structures for fusion categories. As the name suggests, they are categories with enough structures to capture fusion rules. 
\begin{definition}\label{def-rfc}
A category $\cat$ is a \textbf{fusion category} if, 
\begin{itemize}
\item $\cat$ is a semi-simple category over complex numbers, $\mathbb{C}$, 
\item there exist finitely many simple objects in the category, and any non-simple object has a finite direct sum of simple objects, 
\item $\cat$ is monoidal, 
\item $\cat$ is rigid, 
\item for every pair of simple objects $(X_i, X_j)$, $hom(X_i, X_j) \cong \delta_{ij}\mathbb{C}$,
\item the unit object $I$ is simple, we assign index $0$ to the unit object, $X_0 = I$. 
\end{itemize} 
\end{definition}

Being semi-simple over complex numbers, $\mathbb{C}$, implies the category is additive and has simple objects. Considering this condition with monoidality indicates any non-simple object, including the tensor product of two simple ones, can be re-written as a direct sum of simple objects. We refer to each summand, $X_k$, of the direct sum as an outcome of fusion. 
\begin{equation}
	X_i \otimes X_j \cong \bigoplus_{k} N_{ij}^k X_k
\end{equation}
The above equations specify an important property of any anyonic theory, \textbf{fusion rules}, and $N_{ij}^k$ are called \textbf{fusion coefficients}. So one can write, 
\begin{equation}
	N_{ij}^k = dim(hom(X_i \otimes X_j, X_k))
\end{equation}
We also call these hom-sets \textbf{fusion space} and denote them by $V_{ij}^k$.
\[
V_{ij}^k = hom(X_i \otimes X_j, X_k)
\]
In a similar fashion, we call $V_k^{ij}$, \textbf{ decomposition space}.
\[
V_{k}^{ij} = hom(X_k, X_i \otimes X_j)
\]
As mentioned earlier, $N$'s are counting the number of isomorphisms, in this case, creation and annihilation operators, or sometime in the literature, fusion/decomposition channels. Note that if $N_{ij}^k=0$, then anyon $k$ cannot be obtained by fusion $i,j$. 
If one inspects hom-sets of fusion categories, one realises they are either fusion or decomposition spaces because, as mentioned, Schur's lemma indicates any morphism between simple objects is either the zero or identity morphisms. 

Associators and unitors correspondingly can be indexed in a fusion category.  So associators between the tensor product of three consecutive objects will be indexed by four labels, $F_{ijk}^l$.
\begin{equation}
	(X_i \otimes X_j) \otimes X_k \overset{F_{ijk}^l}{\cong} X_i \otimes (X_j \otimes X_k)
\end{equation}
Why four? because as we mentioned earlier, the zero morphism is the only morphism between different simple objects, so the only non-zero morphisms are $F-$matrices indexed as $F_{ijk}^l$ between the same outcome, $X_l$.

The left- and right-unitors are similarly indexed in the following way,   
\begin{align}
	&X_i \otimes X_0 \overset{r_i}{\cong} X_i, & X_0 \otimes X_i \overset{l_i}{\cong} X_i 
\end{align}
A rather not-explicitly-spelt-out property in the literature of anyonic theories is \textbf{skeletality}, i.e, a category which describes an anyonic theory should be a skeleton. Because, in nature, we consider isomorphic anyons identical, we need to work with a skeleton of a fusion category essentially. To find a skeleton of a category, we choose and keep a representative object from each class of isomorphic objects. Therefore, by fusion category, we mean a skeleton of a fusion category such that if there exist two isomorphic objects $X'_i \cong X_i$, we only keep $X_i$ in the category, and we discard isomorphic objects and excessive morphisms. Such  a category is a sub-category of the main category we started from. (Note that we are incidentally assuming the axiom of choice. That should be a legitimate assumption in the context of physical theories.)~\cite{adamek2009}

Another widely assumed property is that we work with a category whose unitors are identity morphisms. Meaning we work with a category for which $r_i = l_i = id_i$. This property subsequently makes $F-$matrices with at least one zero index identity, which reduces the number of pentagonal equations. 

\begin{theorem}
A skeletal fusion category, $(\mathcal{C}, \otimes, F, r, l)$, is equivalent to a skeletal fusion category with strict unitors, $(\mathcal{C}, \otimes', F', id, id)$. 
\end{theorem}
\begin{proof}
For the full proof, check \cite{mythesis}. A proof sketch is as follows; we copy objects and morphisms of the main category and define a new tensor product $\otimes'$ on objects and morphisms. We then need to define an equivalence functor, which in the functor part is identity. Natural isomorphisms of the equivalence functor are defined based on unitors of the starting category, $l$ and $r$. This results in identity unitors $l'$ and $r'$ in the target category.
\end{proof}
\begin{corollary}
In a fusion category with strict unitors, $F-$matrices with at least one zero index are identities. 
\end{corollary}
\begin{proof}
In the triangle equation, Figure \ref{fig:triangle}, let $X=i$, $Y=j$, and $i \otimes j = k$. If we have $l=r=id$, then $(id_i \otimes l_j)F_{i0j}^k = (r_i \otimes id_j)$, which results in $F_{i0j}^k = id_k$. Other variations of $F_{0ij}^k=id_{k}$ and $F_{ij0}^k=id_{k}$ can be obtained by similar procedure. 
\end{proof}
The desired category should have another structure that stems from the exchange statistics of anyons; as mentioned earlier, the anyonic exchange behaviour is captured by the braid group rather than the permutation group. A braid can be defined as a natural transformation between two functors; tensor and opposite-tensor. 

An opposite-tensor functor is a functor which first swaps an ordered pair of objects or morphisms and then tensors them. 
\begin{align}
& \otimes^{opp}: \cat \times \cat \longrightarrow \cat, \\
& \otimes^{opp}(A, B) = \otimes (B, A)  = B \otimes A,\\
& \otimes^{opp}(f, g)= \otimes (g, f) = g \otimes f
\end{align}
\begin{definition}
A fusion category $(\cat, \otimes, F, l, r, I)$ is \textbf{braided}, if there exists a natural transformation 
\[R: \otimes \longrightarrow \otimes^{opp}\] 
which satisfies Hexagonal equations Figure \ref{eq:hexagonal}. 
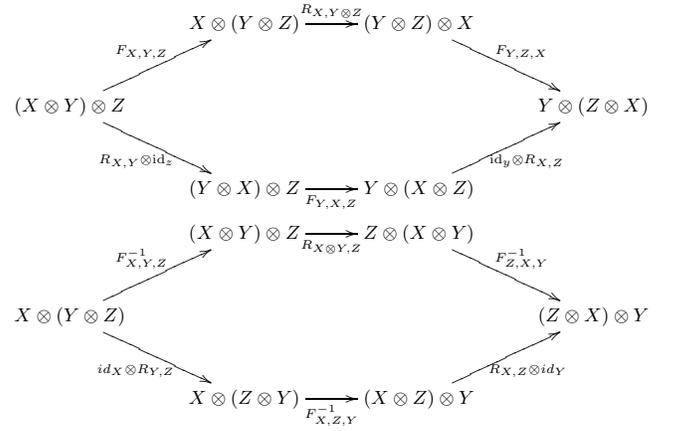
\begin{figure}
\scalebox{0.8}{
		$\begin{array}{c}
			\xymatrix{
				& X \otimes (Y \otimes Z) \ar[r]^{R_{X,Y \otimes Z}} & (Y \otimes Z) \otimes X \ar[dr]^{F_{Y,Z,X}} \\
				(X \otimes Y) \otimes Z \ar[ur]^{F_{X,Y,Z}} \ar[dr]_{R_{X,Y} \otimes \id_z} & & & Y \otimes (Z \otimes X) \\
				& (Y \otimes X) \otimes Z \ar[r]_{F_{Y,X,Z}} & Y \otimes (X \otimes Z) \ar[ur]_{\id_y \otimes R_{X,Z}}
			}
		\end{array}$
	}
 \scalebox{0.8}{
		$\begin{array}{c}
			\xymatrix{
				& (X \otimes Y) \otimes Z \ar[r]_{R_{X \otimes Y, Z}} & Z \otimes (X \otimes Y)  \ar[dr]^{F^{-1}_{Z, X, Y}} \\
				X \otimes (Y \otimes Z) \ar[ur]^{F^{-1}_{X, Y, Z}} \ar[dr]_{id_X \otimes R_{Y, Z}} & & & (Z \otimes X) \otimes Y \\
				& X \otimes (Z \otimes Y) \ar[r]_{F^{-1}_{X, Z, Y}} & (X \otimes Z) \otimes Y \ar[ur]_{R_{X,Z} \otimes id_Y}
			}
		\end{array}$
	}
		\caption{Hexagonal equations.}\label{eq:hexagonal}
	\end{figure}
\end{definition}  
Take the example of a group $G$ as a category; the category is braided only if $G$ is abelian. Note that for a braided tensor category, we have
\begin{align}
	&l_X \circ R_{X, I} = r_X, && r_X \circ R_{I, X} = l_X, && R_{I, X} = R_{X, I}^{-1} &
\end{align}
\begin{rem}
In a category with strict unitors, $R_{0, i}= R_{i, 0}^{-1}=id_i$. 
\end{rem}
Note also the equality $R \otimes R = id$ does not hold in general. It only holds if the category is symmetric and the swapping of particles is captured by the permutation group rather than the braid group. An example of symmetric categories is the category of vector spaces, $\vc$. 

Earlier, it was mentioned braided categories are directly connected to braid groups. Having hexagonal identities, one can prove that the Yang-Baxter equation is obtainable as a corollary of Artin relations of the braid group~\cite{kassel2008}. 
\begin{equation}
	R_{12}R_{13}R_{23} = R_{23}R_{13}R_{12}
\end{equation}
Where $R_{ij}$ keeps the third strand and braids strand numbers $i$ and $j$. We will explicitly show this equation with the ZX-representation of the Ising and Fibonacci models. 

The next structure we need to define is a kind of trace, but before that, we need to define an isomorphism between an object $X$ and its double star $(X^{*})^*$. 
\begin{definition}
A rigid monoidal category equipped with isomorphisms $\phi_X: X \longrightarrow X^{**}$ such that they satisfy the following conditions is \textbf{pivotal}. 
	\begin{align}
		& \phi_{X \otimes Y} = \phi_X \otimes \phi_Y \\
		& f^{**} = f
	\end{align} 
\end{definition}
Having these isomorphisms, we are able to define left and right traces, which are not generally equivalent. We are, however, interested in the categories with equivalent left and right traces. 
\begin{definition}
In a pivotal category, given $f: X \longrightarrow X$, one can define left and right \textbf{traces} as appears below, 
\begin{align}
& tr^r(f) = e_{X^*} \circ (\phi_X \otimes id_{X^*}) \circ (f \otimes id_{X^*}) \circ i_X \\
& tr^l(f) = e_X \circ (e_{X^*} \otimes f) \circ (id_{X^*} \otimes \phi_X^{-1}) \circ i_{X^*}
\end{align}
\end{definition}
\begin{definition}
In a pivotal category, if for every morphism $f$,  $tr^r(f)=tr^l(f)$, the category is called \textbf{spherical}. 
\end{definition}
So far, pivotal and spherical structures are only defined in rigid monoidal categories, but one can further examine the interaction between $\phi$ and braids in a braided rigid monoidal category \cite{etingof2005, turaev2016}. This results in the definition of twist, with a rather interesting physical interpretation. Twists are related to topological spins, which are, in general, different from Dirac spins. Topological spin specifies the phase a particle picks from rotating. It can also be interpreted as a twist in the world-line of a particle; for further discussion on this, refer to Appendix E of~\cite{kitaev2006}. 
\begin{definition}
A braiding is compatible with a pivotal structure if isomorphisms $\theta_X = \psi_X \phi_X$ where $\psi_X = (id_X \otimes e_{X^*}) \circ (R_{X^**, X} \otimes id_{X*}) \circ (id_{X^{**}}\otimes i_X)$ satisfies $\theta_{X^*} = \theta_X^*$. A \textbf{ribbon} category is a spherical braided category with a compatible braiding. 
\end{definition} 
\begin{definition}
A spherical braided fusion category with a compatible braiding is \textbf{ribbon fusion category} (RFC). 
\end{definition}

One of the computationally insightful properties of objects in an RFC is their dimension. For every object $X_i$, we define $d_i = tr(id_i)$. This essentially shows how computational space grows \cite{preskill1998}. An important question to answer is the sign of $d_i$. We show under unitary conditions for an RFC, $d_i \geq 0$. hermitian and unitary definitions agree with definitions that appeared in \cite{turaev2016}. 
\begin{definition}
A \textbf{dagger} functor $\dagger: \cat \longrightarrow \cat^{opp}$ is a functor that keeps objects and reverses the direction of morphisms. 
\[
f: A \longrightarrow B, \hspace{1cm}f^\dagger: B \longrightarrow A 
\]
\end{definition}
\begin{definition}
An RFC is hermitian if for every morphism$f \in hom(X, Y)$ the dagger of every morphism $f^\dagger \in hom(Y, X)$ satisfies below conditions, 
\begin{align}
& (f^\dagger)^\dagger = f, & (f\otimes g)^\dagger = f^\dagger \otimes g^\dagger,& \\
& (f \circ g)^\dagger = g^\dagger \circ f^\dagger, & (id_i)^\dagger = id_i&.
\end{align}
and braid and twist satisfy, 
\begin{align}
& (R_{ij})^\dagger = R_{ij}^{-1}, & (\theta_i)^\dagger = \theta_i^{-1} 
\end{align}
Additionally, dual morphisms are compatible with $\dagger$, 
\begin{align}
& (i_j)^\dagger = e_j \circ R_{jj*} \circ (\theta_i \otimes id_{j*}), \\
& (e_j)^\dagger = (id_{j*} \otimes \theta_j^{-1}) \circ R_{j*j}^{-1} \circ i_j .
\end{align}
\end{definition}
A dagger in an RFC assigns to each morphism in a fusion space, a morphism in the decomposition space, meaning, $\dagger: V_{ij}^k \longrightarrow V_{k}^{ij}$. Or you can also think of it as the dagger of a projection operator is an injection operator. We are now properly positioned to define the inner product on a hermitian RFC. 

\begin{definition}
In a hermitian RFC, an inner product of a pair of morphisms $f, g \in hom(X, Y)$, is defined as, 
\begin{align}
& \langle, \rangle: hom(X, Y) \times hom(X, Y) \longrightarrow \mathbb{C}, \\ 
& \langle f, g\rangle = \frac{1}{\sqrt{dim(X) dim(Y)}} tr(f^\dagger g).
\end{align}
we can check all properties of the inner product. Note we do not explicitly assume $tr(f^\dagger f)$ takes positive values, but if it is positive definite, then the category is \textbf{unitary} \cite{turaev2016}. 
\end{definition}
In a Unitary Ribbon Fusion Category, URFC, one can observe that the dimension of each object $d_i$ is positive. Without the unitarity condition, we could not conclude the following, 
\begin{equation}
	tr(id_{i*} id_i) \geq 0.
\end{equation}
With all these ingredients, we are in the right position to state that a (skeletal) URFC describes an anyonic theory. Simple objects of the category are anyon types, and Schur's lemma guarantees \textbf{selection and superselection rules} hold. Selection rules restrict the possible transitions between quantum states. For instance, it imposes that Anyon $i$,  in the absence of integration, cannot suddenly change to another anyon type. Superselection rules specify which quantum states cannot form a coherent state. For instance, a boson and fermion cannot form a coherent state. In the topological quantum computation context, different anyon types cannot form coherent states~\cite{kitaev2006}. 

Furthermore, fusion rules are defining rules of the theory, $F-$matrices are solutions of pentagonal equations, and braiding matrices, $R$'s, result from solving hexagonal equations. 

An important and undiscussed point is the necessity of modularity condition as a physical requirement for anyonic models. Modular categories ensure the theory has a corresponding topological background field. We did not discuss the modularity condition because we are solely interested in TQC. We refer the interested reader to \cite{bakalov, turaev2016} for further information. 

Therefore, URFC completely captures an anyonic theory's expected properties and structures. However, as we mentioned, the category which describes an anyonic theory is often convoluted with the TQC category. In other words, it is often assumed the category of computation is similar to the category of anyons. Here, we argue that a better perspective to make sense of computation coherent with categorical quantum mechanics is considering TQC as a subcategory of the category of Hilbert spaces, \textbf{Hilb}. 

\section{Categorical Quantum Mechanics}
The idea of categorical quantum mechanics was proposed by Abramsky, and Coecke \cite{abramsky2004}. 
Categorical quantum mechanics essentially attempts to shift the perspective from quantum states to processes. In other words, everything is treated as a process (including states, which are just ``preparation'' processes), and quantum features start to appear as one considers the compositions of processes. For further discussions of this approach and the relationship between this picture and the usual formulation of quantum theory, see e.g.~\cite{bob2021}.  



As we mentioned earlier, CQM aims to formulate quantum structures in terms of categorical ones, namely those of a dagger symmetric monoidal category. Well-developed examples of such a category are the category of relations, \textbf{Rel}, and the category of finite dimensional Hilbert spaces, \textbf{Hilb}. 

The category of finite dimensional Hilbert spaces, \hilb, is a symmetric monoidal category, and as such, it comes with a convenient graphical language called \textit{string diagram notation} for representing morphisms. This notation has its roots in Feynman diagrams and Penrose's graphical notation for tensor contraction~\cite{penrose1971applications}, and was formalised in the 1990s by Joyal and Street for generic monoidal categories~\cite{joyal1991}. A notable feature of this notation is that diagrams that can be deformed into each other (i.e. that are topologically isotopic) describe the same morphism. This topological representation of morphisms in monoidal categories thus gives us a handle on the structure of braids and knots using categorical tools~\cite{turaev2016}.

A central feature of CQM is the reliance on \textit{graphical calculi}. A graphical calculus consists of a set of generating morphisms in a monoidal category, subject to a collection of equational rules that can be used to reason about equality of morphisms.
The most well-known of these is the ZX-calculus, a convenient language for quantum computations over qubits. The ZX-calculus can be seen as a strict superset of circuit notation, in the sense that quantum circuits can be readily interpreted as ZX-diagrams, but there furthermore exist many ZX-diagrams that do not correspond to circuits. 

Its generators can represent arbitrary linear maps between qubits in the category of Hilbert spaces (such as those coming from quantum circuits), and its rules have been used for many applications ranging from quantum circuit optimisation~\cite{duncan2020graphtheoretic,kissinger2019tcount,backens2021there} and measurement-based quantum computing~\cite{DP2,DuncanMBQC} to quantum error correction~\cite{kissinger2022phasefree,horsman2017surgery}. We outline a short summary of CQM and the ZX-calculus here, and in the next section, we explain how TQC fits within this picture. 

To put it more explicitly, a ZX-diagram consists of two kinds of nodes called \textit{Z spiders}~(green/lightly shaded nodes) and \textit{X spiders}~(red/darkly shaded nodes), which can be labelled by an angle, $\alpha \in [0, \pi]$,  and connected by wires, Equations \ref{eq:spiders}.

Concretely Z and X spiders are defined as the following linear maps:
\begin{equation}\label{eq:spiders}
\begin{aligned}
\begin{tikzpicture}[scale=0.5]
	\begin{pgfonlayer}{nodelayer}
		\node [style=Z phase dot] (0) at (0, 0) {$\alpha$};
		\node [style=none] (1) at (1.25, 1) {};
		\node [style=none] (2) at (-1.25, 1) {};
		\node [style=none] (3) at (-1.25, -1) {};
		\node [style=none] (4) at (1.25, -1) {};
		\node [style=none] (5) at (1.25, 0.5) {};
		\node [style=none] (6) at (-1.25, 0.5) {};
		\node [style=none, rotate=90] (7) at (-1, -0.25) {...};
		\node [style=none, rotate=90] (8) at (1, -0.25) {...};
	\end{pgfonlayer}
	\begin{pgfonlayer}{edgelayer}
		\draw [in=-141, out=0, looseness=0.75] (3.center) to (0);
		\draw [in=180, out=-39, looseness=0.75] (0) to (4.center);
		\draw [in=180, out=22, looseness=0.75] (0) to (5.center);
		\draw [in=180, out=39, looseness=0.75] (0) to (1.center);
		\draw [in=0, out=158, looseness=0.75] (0) to (6.center);
		\draw [in=141, out=0, looseness=0.75] (2.center) to (0);
	\end{pgfonlayer}
\end{tikzpicture} \ &:= \ \ketbra{0...0}{0...0} +
	e^{i \alpha} \ketbra{1...1}{1...1} \\
	\begin{tikzpicture}[scale=0.5]
		\begin{pgfonlayer}{nodelayer}
			\node [style=X phase dot] (0) at (0, 0) {$\alpha$};
			\node [style=none] (1) at (1.25, 1) {};
			\node [style=none] (2) at (-1.25, 1) {};
			\node [style=none] (3) at (-1.25, -1) {};
			\node [style=none] (4) at (1.25, -1) {};
			\node [style=none] (5) at (1.25, 0.5) {};
			\node [style=none] (6) at (-1.25, 0.5) {};
			\node [style=none, rotate=90] (7) at (-1, -0.25) {...};
			\node [style=none, rotate=90] (8) at (1, -0.25) {...};
		\end{pgfonlayer}
		\begin{pgfonlayer}{edgelayer}
			\draw [in=-141, out=0, looseness=0.75] (3.center) to (0);
			\draw [in=180, out=-39, looseness=0.75] (0) to (4.center);
			\draw [in=180, out=22, looseness=0.75] (0) to (5.center);
			\draw [in=180, out=39, looseness=0.75] (0) to (1.center);
			\draw [in=0, out=158, looseness=0.75] (0) to (6.center);
			\draw [in=141, out=0, looseness=0.75] (2.center) to (0);
		\end{pgfonlayer}
	\end{tikzpicture}\ &:= \ \ketbra{+...+}{+...+} +
	e^{i \alpha} \ketbra{-...-}{-...-}
\end{aligned}
\end{equation}
Here, we have adopted ``bra-ket'' notation from quantum computing, where $|0\>$ and $|1\>$ represent the standard basis for $\mathbb C^2$, $|{+}\> = 1/\sqrt{2} (|0\> + |1\>)$ and $|{-}\> = 1/\sqrt{2}(|0\> - |1\>)$ the Hadamard (a.k.a. ``plus'') basis, $\<0|, \<1|, \<{+}|, \<{-}|$ their associated conjugate-transposes, and juxtaposition represents the tensor product of basis vectors.

\textit{ZX diagrams} are formed from spiders much like quantum circuits are built from gates: plugging two diagrams together sequentially represents composition of the associated linear maps and stacking diagrams on top of each other represents tensor product. Crossing wires represent the ``swap'' map $\mathbb C^2 \otimes \mathbb C^2 \to \mathbb C^2 \otimes \mathbb C^2$ which maps $|\psi\> \otimes |\phi\> \mapsto |\phi\> \otimes |\psi\>$.


ZX diagrams obey a set of diagram rewrite rules called the \textit{ZX calculus}. One particular presentation of these rules is given in Figure~\ref{fig:zx-rules}.  The yellow box represents Hadamard matrix, which interchanges the two bases $\{|0\>,|1\>\}$ and $\{|{+}\>,|{-}\>\}$.

\begin{figure}
\centering
\scalebox{0.8}{
\begin{tikzpicture}[scale=0.5]
	\begin{pgfonlayer}{nodelayer}
		\node [style=Z phase dot] (0) at (-8, 3.5) {$\beta$};
		\node [style=none] (1) at (-7, 4) {};
		\node [style=none] (2) at (-9.75, 4) {};
		\node [style=none] (3) at (-9.75, 2.75) {};
		\node [style=none] (4) at (-7, 2.75) {};
		\node [style=none] (5) at (-7, 3.75) {};
		\node [style=none] (6) at (-9.75, 3.75) {};
		\node [style=none, rotate=90] (7) at (-10, 3.25) {...};
		\node [style=none, rotate=90] (8) at (-7, 3.25) {...};
		\node [style=none] (9) at (-10.5, 5.25) {};
		\node [style=Z phase dot] (10) at (-9.5, 5) {$\alpha$};
		\node [style=none] (11) at (-7.5, 5.5) {};
		\node [style=none] (12) at (-10.5, 5.5) {};
		\node [style=none] (13) at (-7.5, 4.25) {};
		\node [style=none, rotate=90] (14) at (-7.5, 4.75) {...};
		\node [style=none] (15) at (-7.5, 5.25) {};
		\node [style=none] (16) at (-10.5, 4.25) {};
		\node [style=none, rotate=90] (17) at (-10.25, 4.75) {...};
		\node [style=none] (18) at (-7, 4.25) {};
		\node [style=none] (19) at (-7, 5.5) {};
		\node [style=none] (20) at (-7, 5.25) {};
		\node [style=none] (21) at (-10.5, 3.75) {};
		\node [style=none] (22) at (-10.5, 2.75) {};
		\node [style=none] (23) at (-10.5, 4) {};
		\node [style=none] (24) at (-5.75, 4.25) {$\overset{=}{\text{f-rule}}$};
		\node [style=none, rotate=45] (25) at (-8.75, 4.25) {...};
		\node [style=none] (26) at (-4.5, 5.5) {};
		\node [style=none] (27) at (-1, 5) {};
		\node [style=none, rotate=90] (28) at (-4, 4) {...};
		\node [style=none] (29) at (-1, 2.75) {};
		\node [style=none, rotate=90] (30) at (-1.25, 4) {...};
		\node [style=none] (31) at (-4.5, 5) {};
		\node [style=Z phase dot] (32) at (-2.75, 4.25) {$\ \alpha\!+\!\beta\ $};
		\node [style=none] (33) at (-1, 5.5) {};
		\node [style=none] (34) at (-4.5, 2.75) {};
		\node [style=Z phase dot] (36) at (-3.25, -0.25) {$\ (\textrm{-}1)^a \alpha\ $};
		\node [style=none] (37) at (-1, 1.5) {};
		\node [style=none] (38) at (-5.75, -0.25) {$\overset{=}{\pi\text{-rule}}$};
		\node [style=none] (39) at (-4.5, -0.25) {};
		\node [style=none] (40) at (-1, -1.5) {};
		\node [style=X phase dot] (41) at (-1.75, 0.5) {$a\pi$};
		\node [style=X phase dot] (42) at (-1.75, -1.5) {$a\pi$};
		\node [style=X phase dot] (43) at (-9.25, -0.25) {$a \pi$};
		\node [style=none] (44) at (-6.75, 1.25) {};
		\node [style=Z phase dot] (45) at (-8.25, -0.25) {$\alpha$};
		\node [style=none, rotate=90] (46) at (-7, -0.5) {...};
		\node [style=none] (47) at (-1, 0.5) {};
		\node [style=none, rotate=90] (48) at (-1.75, -0.5) {...};
		\node [style=none] (49) at (-6.75, 0.25) {};
		\node [style=none] (50) at (-10, -0.25) {};
		\node [style=none] (51) at (-6.75, -1.5) {};
		\node [style=X phase dot] (52) at (-1.75, 1.5) {$a\pi$};
		\node [style=X phase dot] (56) at (8, -3.5) {$a\pi$};
		\node [style=Z phase dot] (57) at (3, -4.25) {$\alpha$};
		\node [style=none] (58) at (4.25, -5) {};
		\node [style=none] (59) at (5.5, -4.25) {$=$};
		\node [style=none] (60) at (9, -5) {};
		\node [style=none] (61) at (4.25, -3.5) {};
		\node [style=none] (64) at (9, -3.5) {};
		\node [style=X phase dot] (65) at (8, -5) {$a\pi$};
		\node [style=X phase dot] (68) at (2, -4.25) {$a\pi$};
		\node [style=hadamard] (69) at (8.5, 5.25) {};
		\node [style=Z phase dot] (70) at (2.75, 4) {$\alpha$};
		\node [style=none] (71) at (6.5, 4) {};
		\node [style=none, rotate=90] (72) at (8.5, 3.75) {...};
		\node [style=none] (73) at (5.5, 4) {$=$};
		\node [style=hadamard] (74) at (8.5, 3) {};
		\node [style=none] (75) at (1.25, 4) {};
		\node [style=X phase dot] (76) at (7.5, 4) {$\alpha$};
		\node [style=none] (77) at (9, 5.25) {};
		\node [style=hadamard] (78) at (2, 4) {};
		\node [style=none] (79) at (9, 4.5) {};
		\node [style=hadamard] (80) at (8.5, 4.5) {};
		\node [style=none] (81) at (4.25, 3) {};
		\node [style=none] (82) at (9, 3) {};
		\node [style=none, rotate=90] (83) at (4, 3.75) {...};
		\node [style=none] (84) at (4.25, 5.25) {};
		\node [style=none] (85) at (4.25, 4.5) {};
		\node [style=Z dot] (87) at (3.5, 0.75) {};
		\node [style=none] (88) at (2.75, 0.75) {};
		\node [style=none] (89) at (4.25, 0.75) {};
		\node [style=none] (90) at (6.75, 0.75) {};
		\node [style=none] (91) at (7.75, 0.75) {};
		\node [style=none] (93) at (5.5, 0.75) {$=$};
		\node [style=none] (94) at (5.5, -1.25) {$=$};
		\node [style=none] (96) at (7.75, -1.25) {};
		\node [style=none] (97) at (4.25, -1.25) {};
		\node [style=none] (98) at (2.75, -1.25) {};
		\node [style=none] (99) at (6.75, -1.25) {};
		\node [style=hadamard] (100) at (3.25, -1.25) {};
		\node [style=hadamard] (101) at (3.75, -1.25) {};
		\node [style=X dot] (103) at (-1.75, -5) {};
		\node [style=none] (104) at (-1, -3.5) {};
		\node [style=X dot] (105) at (-8.75, -4.25) {};
		\node [style=Z dot] (106) at (-7.5, -4.25) {};
		\node [style=none] (107) at (-4, -5) {};
		\node [style=none] (108) at (-9.5, -3.5) {};
		\node [style=none] (109) at (-4, -3.5) {};
		\node [style=Z dot] (110) at (-3.25, -3.5) {};
		\node [style=none] (111) at (-6.75, -3.5) {};
		\node [style=none] (112) at (-5.75, -4.25) {$=$};
		\node [style=none] (113) at (-9.5, -5) {};
		\node [style=X dot] (114) at (-1.75, -3.5) {};
		\node [style=none] (115) at (-1, -5) {};
		\node [style=Z dot] (116) at (-3.25, -5) {};
		\node [style=none] (117) at (-6.75, -5) {};
		\node [style=none] (118) at (-3.75, -1.5) {$e^{ia\alpha}$};
		\node [style=none] (119) at (6.625, -4.25) {{ $\frac{e^{ia\alpha}}{\sqrt{2}}$ }};
		\node [style=none] (120) at (-4.5, -4.25) {{\scriptsize $\sqrt{2}$}};
	\end{pgfonlayer}
	\begin{pgfonlayer}{edgelayer}
		\draw [style=simple, in=-141, out=0, looseness=0.75] (3.center) to (0);
		\draw [style=simple, in=180, out=-39, looseness=0.75] (0) to (4.center);
		\draw [style=simple, in=180, out=22, looseness=0.75] (0) to (5.center);
		\draw [style=simple, in=180, out=39, looseness=0.75] (0) to (1.center);
		\draw [style=simple, in=0, out=180] (0) to (6.center);
		\draw [style=simple, in=165, out=0] (2.center) to (0);
		\draw [style=simple, in=-141, out=0, looseness=0.75] (16.center) to (10);
		\draw [style=simple, in=180, out=-15] (10) to (13.center);
		\draw [style=simple, in=180, out=22] (10) to (15.center);
		\draw [style=simple, in=180, out=39, looseness=0.75] (10) to (11.center);
		\draw [style=simple, in=0, out=158, looseness=0.75] (10) to (9.center);
		\draw [style=simple, in=141, out=0, looseness=0.75] (12.center) to (10);
		\draw [style=simple, bend left] (10) to (0);
		\draw [style=simple] (11.center) to (19.center);
		\draw [style=simple] (15.center) to (20.center);
		\draw [style=simple] (13.center) to (18.center);
		\draw [style=simple] (23.center) to (2.center);
		\draw [style=simple] (21.center) to (6.center);
		\draw [style=simple] (22.center) to (3.center);
		\draw [style=simple, bend left] (0) to (10);
		\draw [style=simple, in=-120, out=0] (34.center) to (32);
		\draw [style=simple, in=180, out=-60] (32) to (29.center);
		\draw [style=simple, in=180, out=45, looseness=0.75] (32) to (27.center);
		\draw [style=simple, in=180, out=60] (32) to (33.center);
		\draw [style=simple, in=0, out=135, looseness=0.75] (32) to (31.center);
		\draw [style=simple, in=120, out=0] (26.center) to (32);
		\draw [style=simple, in=180, out=-60, looseness=0.75] (45) to (51.center);
		\draw [style=simple, in=180, out=45, looseness=0.75] (45) to (49.center);
		\draw [style=simple, in=180, out=75, looseness=0.75] (45) to (44.center);
		\draw [style=simple, in=180, out=0, looseness=0.50] (43) to (45);
		\draw [style=simple] (50.center) to (43);
		\draw [style=simple, in=-60, out=180, looseness=0.75] (42) to (36);
		\draw [style=simple, in=0, out=180, looseness=0.75] (36) to (39.center);
		\draw [style=simple, in=180, out=30] (36) to (41);
		\draw [style=simple, in=60, out=180, looseness=0.75] (52) to (36);
		\draw [style=simple] (37.center) to (52);
		\draw [style=simple] (47.center) to (41);
		\draw [style=simple] (40.center) to (42);
		\draw [style=simple, in=180, out=-60, looseness=0.75] (57) to (58.center);
		\draw [style=simple, in=180, out=60, looseness=0.75] (57) to (61.center);
		\draw [style=simple, in=180, out=0, looseness=0.50] (68) to (57);
		\draw [style=simple] (64.center) to (56);
		\draw [style=simple] (60.center) to (65);
		\draw [style=simple, in=180, out=-60, looseness=0.75] (70) to (81.center);
		\draw [style=simple, in=180, out=45, looseness=0.75] (70) to (85.center);
		\draw [style=simple, in=180, out=75, looseness=0.75] (70) to (84.center);
		\draw [style=simple, in=180, out=0, looseness=0.75] (78) to (70);
		\draw [style=simple] (75.center) to (78);
		\draw [style=simple, in=-60, out=180, looseness=0.75] (74) to (76);
		\draw [style=simple, in=0, out=180, looseness=0.75] (76) to (71.center);
		\draw [style=simple, in=180, out=45, looseness=0.75] (76) to (80);
		\draw [style=simple, in=75, out=180, looseness=0.75] (69) to (76);
		\draw [style=simple] (77.center) to (69);
		\draw [style=simple] (79.center) to (80);
		\draw [style=simple] (82.center) to (74);
		\draw (88.center) to (89.center);
		\draw (90.center) to (91.center);
		\draw (98.center) to (97.center);
		\draw (99.center) to (96.center);
		\draw [style=simple] (116) to (114);
		\draw [style=simple] (103) to (110);
		\draw [style=simple] (110) to (114);
		\draw [style=simple] (116) to (103);
		\draw [style=simple] (115.center) to (103);
		\draw [style=simple] (107.center) to (116);
		\draw [style=simple] (110) to (109.center);
		\draw [style=simple] (114) to (104.center);
		\draw [style=simple, bend right] (113.center) to (105);
		\draw [style=simple, bend right] (105) to (108.center);
		\draw [style=simple] (105) to (106);
		\draw [style=simple, bend right] (106) to (117.center);
		\draw [style=simple, bend left] (106) to (111.center);
	\end{pgfonlayer}
\end{tikzpicture}}
\caption{
A convenient presentation for the ZX-calculus. These rules hold
for all $\alpha, \beta \in [0, 2 \pi)$ and $a\in\{0,1\}$. They also hold with the colours (red and green) interchanged and with the inputs and outputs permuted arbitrarily. Note all spiders are labeled by angles, i.e. real numbers taken modulo $2\pi$, and a spider with no label is assumed to have an angle of $0$.}\label{fig:zx-rules}
\end{figure}
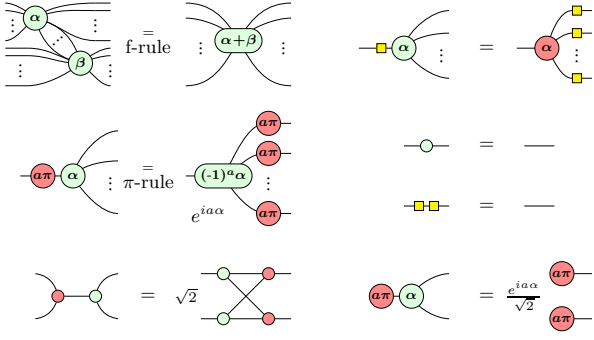

These rules can be applied while transforming diagrams back and forth to each other. While we won't be particularly interested in applying all of the rules in the calculus here, we specifically use ``f-rule,'' also known as spider fusion and ``$\pi$-rule''. ``f-rule'' adds phases of spiders with the same colour connected via wires, and introduces another spider with the same colour and addition of phases. If a $\pi$-spider is adjacent to an opposite colour spider,``$\pi$-rule,'' states the $\pi$-spider passes through the spider but changes the phase. We use a combination of these rules in Section~\ref{sec:braid-relations}.

We are furthermore particularly interested in the so-called ``P-rule'', which has appeared in proposed extensions to the rules in Fig.~\ref{fig:zx-rules} such as the complete graphical calculus proposed in \cite{vilmart2019near}, due to its relationship to topological quantum computing, which we will explain in the next section. 

Before giving the P-rule in generality, we will give a special case, which relates two equivalent representations of the Hadamard matrix in terms of three consecutive Z and X spiders with $\pi/2$ angles. 

\begin{definition}(Hadamard rule)
A Hadamard matrix has an equivalent ZX-representation of three $\frac{\pi}{2}$ (Z-X-Z) or (X-Z-X) spiders:
	\begin{equation}\label{eq:had-rule}
	H \ =\ 
\begin{tikzpicture}[scale=0.5]
	\begin{pgfonlayer}{nodelayer}
		\node [style=Z phase dot] (0) at (-4, 0) {$\frac\pi2$};
		\node [style=X phase dot] (1) at (-3, 0) {$\frac\pi2$};
		\node [style=Z phase dot] (2) at (-2, 0) {$\frac\pi2$};
		\node [style=none] (3) at (-5, 0) {};
		\node [style=none] (4) at (-1, 0) {};
		\node [style=none] (5) at (0, 0) {$=$};
		\node [style=none] (6) at (1, 0) {};
		\node [style=X phase dot] (7) at (2, 0) {$\frac\pi2$};
		\node [style=X phase dot] (8) at (4, 0) {$\frac\pi2$};
		\node [style=Z phase dot] (9) at (3, 0) {$\frac\pi2$};
		\node [style=none] (10) at (5, 0) {};
	\end{pgfonlayer}
	\begin{pgfonlayer}{edgelayer}
		\draw (0) to (2);
		\draw (3.center) to (0);
		\draw (2) to (4.center);
		\draw (6.center) to (10.center);
	\end{pgfonlayer}
\end{tikzpicture}
	\end{equation}
\end{definition}

This ability to swap a Z-X-Z chain of spiders for an X-Z-X chain holds not just for $\pi/2$ angles, but in general. However, the expression becomes more complicated.

\begin{definition}($P$-rule)
	Given a composition of (Z-X-Z) spiders, by using relation between angles below, one can find an equivalent combination of (X-Z-X) spiders and vice versa. 
	\begin{equation}
\begin{tikzpicture}[scale=0.5]
	\begin{pgfonlayer}{nodelayer}
		\node [style=Z phase dot] (0) at (-4, 0) {$\alpha$};
		\node [style=X phase dot] (1) at (-3, 0) {$\beta$};
		\node [style=Z phase dot] (2) at (-2, 0) {$\gamma$};
		\node [style=none] (3) at (-5, 0) {};
		\node [style=none] (4) at (-1, 0) {};
		\node [style=none] (5) at (0, 0) {$=$};
		\node [style=none] (6) at (1, 0) {};
		\node [style=X phase dot] (7) at (2, 0) {$\alpha'$};
		\node [style=X phase dot] (8) at (4, 0) {$\gamma'$};
		\node [style=Z phase dot] (9) at (3, 0) {$\beta'$};
		\node [style=none] (10) at (5, 0) {};
	\end{pgfonlayer}
	\begin{pgfonlayer}{edgelayer}
		\draw (0) to (2);
		\draw (3.center) to (0);
		\draw (2) to (4.center);
		\draw (6.center) to (10.center);
	\end{pgfonlayer}
\end{tikzpicture}
	\end{equation}
	\begin{align}
		& z = cos(\frac{\beta}{2})cos(\frac{\alpha+\gamma}{2})+isin(\frac{\beta}{2})cos(\frac{\alpha-\gamma}{2}) \\
		& z' = cos(\frac{\beta}{2})sin(\frac{\alpha+\gamma}{2})-isin(\frac{\beta}{2})sin(\frac{\alpha-\gamma}{2}) 
	\end{align}
	The equivalent three angles are:
	
	\[ (\alpha', \beta', \gamma')= \begin{cases} 
		\alpha' = arg(z)+arg(z')\\
		\beta'= 2arg(\frac{|z|}{|z'|}+i)\\
		\gamma'=  arg(z)-arg(z')
	\end{cases}
	\]
\end{definition}
\section{The Relationship Between TQC and CQM}
In Section~\ref{sec:level2}, we discussed the properties and structures of unitary ribbon fusion categories. Also, we briefly touched on the fact that the category for TQC is sometimes convoluted with the category describing anyonic theories. Given the structure of categorical quantum computing, the question of interaction between TQC and CQM arises. We show TQC category is a subcategory of \hilb.

Knowing this fact, another exciting area to explore is the representation of TQC models with the $ZX$-calculus.  We show this explicitly for the Fibonacci and Ising models, introduce new relations, and conclude that while the Ising model is only the Clifford fragment of the $ZX$-calculus, the Fibonacci model represents a new fragment. This further results in the introduction of a new P-rule. 

Let us review the elements of TQC. For quantum computation, we need a well-defined anyonic theory. This means we have a set of anyon types, fusion rules which essentially specify the outcomes of the tensor product between each pair of objects, and solutions of pentagonal and hexagonal equations. 

The computational space is not an anyonic space; instead, it is the fusion space,$V^k_{ij}$. Physically, we encode qubits in the processes between anyons. So if the outcome of $(X_i \otimes X_j) \otimes X_k$ has two possibilities, either with a single outcome with a multiplicity coefficient $N_{ij}^k=2$ or with more than one outcome, then one can take $V_{ij}^k$ as the  computation space. 

Having this picture in mind, then solutions of pentagonal equations, which we call F-matrices, are represented as matrices and linear isomorphisms between two isomorphic vector spaces $V^l_{(ij)k}\cong V^l_{i(jk)}$. However, the building blocks of the spaces where we perform computations are fusion spaces $V^k_{ij}$, so any other space has an equivalent direct sum of these fusion spaces. 
\begin{equation}
	V_{(ij)k}^l \cong \bigoplus_{ek} V_{ij}^e \otimes V_{ek}^l
\end{equation}
Similar to F-matrices, it should not be surprising that hexagonal equations' solutions, which we call R-matrices, are linear isomorphisms between two isomorphic spaces $V_{ij}^k \cong V_{ji}^k$. With these information and physical properties in mind, in the following we define the TQC category, 

\begin{definition} \label{def-TQC}
Given a well-defined anyonic theory or category, a TQC category consists of the following objects and morphisms:
\begin{itemize}
\item \textbf{Objects} direct sums and tensor products of fusion spaces $V_{ij}^k$. 
\item \textbf{Morphisms} $R$ matrices, $F$ matrices, compositions and tensor products thereof.
\end{itemize}
\end{definition} 
Pentagonal and hexagonal equations can be rewritten by substituting anyons with fusion spaces; one must solve these equations to find F and R solutions explicitly. There are some simplifying rules which reduce the number of equations, for example, in the Fibonacci case from 32 to 1. Because every $F-$matrices has 5 indices and each index can be either Fibonacci anyon or vacuum, but any matrix with at least one vacuum index is identity, so only one non-trivial equation remains to solve, for further information, see \cite{mythesis}. This category is furthermore semi-simple: simple objects are fusion spaces, and any other space is a direct sum of fusion spaces. 
\begin{equation}
	V_{i_1..i_n}^j \cong \bigoplus_{k_1...k_{n-1}} V_{i_1i_2}^{k_1} \otimes  V_{k_1i_3}^{k_2}\otimes  ... \otimes V_{k_{n-1}i_n}^{j}
\end{equation}

Note that Definition~\ref{def-TQC} gives us a computational category different from a general fusion category, Definition~\ref{def-rfc}. Having this definition, we work with objects of a category rather than hom-sets as computational spaces. The latter has been the definition of an anyonic category as well as a TQC category in literature; see, for example, Section 7.1 of \cite{wang2010}.  

Because fusion spaces are hom-sets of a URFC, they come equipped with a well-defined inner product, transforming them into Hilbert spaces. So the relationship between CQM and TQC should be evident by now. Given the definition above, this category is also closed under the tensor product, and the direct sum and the F- and R-matrices are well-behaved. Therefore, the natural next step is to develop the ZX-representation of TQC. We show our attempt at two well-studied models, Ising and Fibonacci. With our definition of a TQC category, it should be evident by now that it is indeed a subcategory of \hilb, and the functor is an inclusion. 

\subsection{The ZX-Representation of TQC}
To build a computational space, we take non-abelian anyons whose fusion outcomes have more than one outcome or an outcome with a fusion coefficient $N_{ij}^k \geq 1$. In both cases here, i.e. Ising and Fibonacci, we take three Ising or Fibonacci anyons, let's call them $a$, and form the computational space as follows: 
\[
V_{aaa}^a = \oplus_x V_{aa}^x \otimes V_{xa}^a
\]
We then construct all unitary operations from associativity and braiding matrices. The $F$-matrix resulting from solving the pentagonal equations changes the basis. 
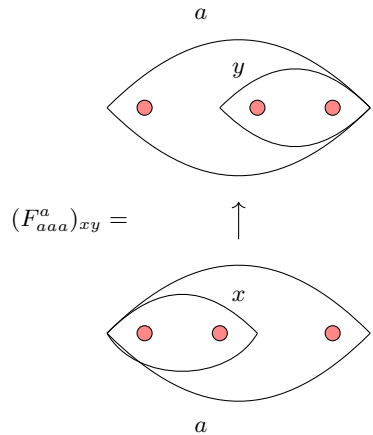
\begin{figure}[!h]
\begin{tikzpicture}
	\begin{pgfonlayer}{nodelayer}
		\node [style=X dot] (0) at (-2.5, 0) {};
		\node [style=X dot] (1) at (-1.5, 0) {};
		\node [style=X dot] (2) at (0, 0) {};
		\node [style=none] (5) at (0.5, 0) {};
		\node [style=none] (6) at (-3, 0) {};
		\node [style=X dot] (7) at (-2.5, 3) {};
		\node [style=X dot] (8) at (-1, 3) {};
		\node [style=X dot] (9) at (0, 3) {};
		\node [style=none] (12) at (0.5, 3) {};
		\node [style=none] (13) at (-3, 3) {};
		\node [style=none] (16) at (-3.5, 1.5) {};
		\node [style=none] (17) at (-3.5, 1.5) {$(F^a_{aaa})_{xy}=$};
		\node [style=none] (18) at (-1.25, 0.5) {$x$};
		\node [style=none] (19) at (-1.25, 3.5) {$y$};
		\node [style=none] (20) at (-1.75, 4.25) {};
		\node [style=none] (21) at (-1.75, -1.25) {};
		\node [style=none] (22) at (-1.75, 4.25) {$a$};
		\node [style=none] (23) at (-1.75, -1.25) {$a$};
		\node [style=none] (26) at (-1.5, 3) {};
		\node [style=none] (27) at (-1, 0) {};
		\node [style=none] (28) at (-1.25, 1.25) {};
		\node [style=none] (29) at (-1.25, 1.75) {};
	\end{pgfonlayer}
	\begin{pgfonlayer}{edgelayer}
		\draw [bend left=45, looseness=1.25] (6.center) to (5.center);
		\draw [in=-135, out=-45, looseness=1.25] (6.center) to (5.center);
		\draw [bend left=45, looseness=1.25] (13.center) to (12.center);
		\draw [in=-135, out=-45, looseness=1.25] (13.center) to (12.center);
		\draw [in=135, out=45, looseness=1.25] (26.center) to (12.center);
		\draw [bend right=45, looseness=1.25] (26.center) to (12.center);
		\draw [bend left=45, looseness=1.25] (6.center) to (27.center);
		\draw [bend right=60] (6.center) to (27.center);
		\draw [style=diredge] (28.center) to (29.center);
	\end{pgfonlayer}
\end{tikzpicture}
\end{figure}\\
Braiding matrices braid or swap anyons. $R_1$'s swap the first two anyons and $R_2$'s the last two anyons. 
\begin{figure}[!h]
\begin{tikzpicture}
	\begin{pgfonlayer}{nodelayer}
		\node [style=X dot] (0) at (-2, 0) {};
		\node [style=X dot] (1) at (-1, 0) {};
		\node [style=X dot] (2) at (0, 0) {};
		\node [style=none] (3) at (-2.5, 0) {};
		\node [style=none] (4) at (-0.5, 0) {};
		\node [style=none] (5) at (0.5, 0) {};
		\node [style=none] (6) at (-3, 0) {};
		\node [style=X dot] (7) at (-2, 3) {};
		\node [style=X dot] (8) at (-1, 3) {};
		\node [style=X dot] (9) at (0, 3) {};
		\node [style=none] (10) at (-2.5, 3) {};
		\node [style=none] (11) at (-0.5, 3) {};
		\node [style=none] (12) at (0.5, 3) {};
		\node [style=none] (13) at (-3, 3) {};
		\node [style=none] (14) at (-1.5, 1.25) {};
		\node [style=none] (15) at (-1.75, 1.5) {};
		\node [style=none] (16) at (-3.5, 1.5) {};
		\node [style=none] (17) at (-3, 1.5) {$R_{1}=$};
		\node [style=none] (18) at (-1, -0.5) {$x$};
		\node [style=none] (19) at (-1.25, 3.75) {$x$};
	\end{pgfonlayer}
	\begin{pgfonlayer}{edgelayer}
		\draw [in=135, out=45] (3.center) to (4.center);
		\draw [bend right=45] (3.center) to (4.center);
		\draw [bend left=45, looseness=1.25] (6.center) to (5.center);
		\draw [in=-135, out=-45, looseness=1.25] (6.center) to (5.center);
		\draw [in=135, out=45] (10.center) to (11.center);
		\draw [bend right=45] (10.center) to (11.center);
		\draw [bend left=45, looseness=1.25] (13.center) to (12.center);
		\draw [in=-135, out=-45, looseness=1.25] (13.center) to (12.center);
		\draw [in=-90, out=90] (0) to (8);
		\draw [bend right, looseness=1.25] (1) to (14.center);
		\draw [bend left=15] (15.center) to (7);
	\end{pgfonlayer}
\end{tikzpicture}
\begin{tikzpicture}
	\begin{pgfonlayer}{nodelayer}
		\node [style=X dot] (0) at (-2.25, 0) {};
		\node [style=X dot] (1) at (-1, 0) {};
		\node [style=X dot] (2) at (0, 0) {};
		\node [style=none] (5) at (0.5, 0) {};
		\node [style=none] (6) at (-3, 0) {};
		\node [style=X dot] (7) at (-2.25, 3) {};
		\node [style=X dot] (8) at (-1, 3) {};
		\node [style=X dot] (9) at (0, 3) {};
		\node [style=none] (12) at (0.5, 3) {};
		\node [style=none] (13) at (-3, 3) {};
		\node [style=none] (17) at (-2.25, 1.5) {$R_{2}=$};
		\node [style=none] (18) at (-1, -0.75) {$y$};
		\node [style=none] (19) at (-1.25, 3.5) {$y$};
		\node [style=none] (26) at (-1.5, 3) {};
		\node [style=none] (27) at (-1.5, 0) {};
		\node [style=none] (28) at (-0.25, 1.25) {};
		\node [style=none] (29) at (-0.5, 1.5) {};
	\end{pgfonlayer}
	\begin{pgfonlayer}{edgelayer}
		\draw [bend left=45, looseness=1.25] (6.center) to (5.center);
		\draw [in=-135, out=-45, looseness=1.25] (6.center) to (5.center);
		\draw [bend left=45, looseness=1.25] (13.center) to (12.center);
		\draw [in=-135, out=-45, looseness=1.25] (13.center) to (12.center);
		\draw [in=135, out=45, looseness=1.25] (26.center) to (12.center);
		\draw [bend right=45, looseness=1.25] (26.center) to (12.center);
		\draw [in=135, out=45, looseness=1.25] (27.center) to (5.center);
		\draw [in=-120, out=-60] (27.center) to (5.center);
		\draw [in=-75, out=75] (1) to (9);
		\draw [in=-30, out=60, looseness=0.75] (2) to (28.center);
		\draw [bend left, looseness=0.75] (29.center) to (8);
	\end{pgfonlayer}
\end{tikzpicture}
\end{figure}
It should be clear that $R_2$ can be constructed as a composition of an $F$-matrix followed by an $R$-matrix and followed by an $F^{-1}$-matrix to return to the first basis we started from. 
\[
R_2 = F^{-1} R_1 F
\]
Early results such as \cite{freedman2002modular} specify the universal classes of the TQC model. Later papers such as \cite{bonesteel2005braid, fan2010braid} suggested a more constructive approach. We restrict ourselves to Ising and Fibonacci anyons and, in a constructive approach, build the ZX-representation of their braid groups up to six strands. 

\subsubsection{Ising Model}
Ising anyons, also known as Majorana fermions, are the simplest realization of non-Abelian anyons \cite{sarma2015majorana}. The Ising model includes two non-trivial anyons $\{\sigma,\psi\}$ and the non-trivial fusion rules are as below:
\begin{align}
	& \sigma \otimes \sigma = 1 \oplus \psi, && \sigma \otimes \psi = \sigma, && \psi \otimes \psi = 1 &
\end{align}
where $1$ is the trivial anyon.
As should be clear from these rules, $\sigma$ is the non-Abelian anyon we use for encoding. It is called Ising or Majorana fermion because of the similarity of Ising anyons statistics with Majorana fermions. 
To encode qubits, we need three Ising anyons so that based on their internal charges, we can map them to qubits, $\{|1\psi \rangle = |0\rangle,|\psi \sigma \rangle = |1 \rangle \}$. 

Fixing this basis, the computational Hilbert space is $H= \oplus_x V_{\sigma \sigma}^x \otimes V_{x \sigma}^\sigma$. One can solve pentagonal and hexagonal equations to find $F$ and $R$ matrices \cite{wang2010}, 
\begin{align}
& F^{isg}=H=\frac{1}{\sqrt{2}} \begin{pmatrix}
		1 & 1 \\ 1 & -1 
	\end{pmatrix}, \\
& R_1^{isg}=-e^{\frac{\pi i}{8}} \begin{pmatrix}
		1 & 0 \\ 0 & -i 
	\end{pmatrix},\\
& R_2^{isg}=-\frac{e^{\frac{-\pi i}{8}}}{\sqrt{2}} \begin{pmatrix}
		1 & i \\ i & 1 
\end{pmatrix},&
\end{align}
Comparing with equation~\eqref{eq:spiders}, we see that  $R^{isg}_1$ matrix is a Z-spider with one input, one output, and a phase angle of $\frac{-\pi}{2}$. We mentioned before that Hadamard also has a $ZX$-representation as a combination of three $\frac{\pi}{2}$ angles. So $R_2^{isg}=HR_1^{isg}H$ gives the other braid generator~\ref{fig:zx-of-ising}. 
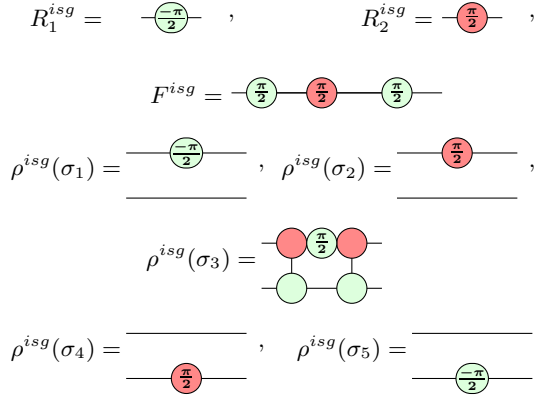
\begin{figure}
\begin{tikzpicture}[scale=0.8]
	\begin{pgfonlayer}{nodelayer}
		\node [style=none] (0) at (-7, 2) {$R_{1}^{isg}=$};
		\node [style=none] (1) at (-5.75, 2) {};
		\node [style=none] (2) at (-5.25, 2) {};
		\node [style=none] (3) at (-4.75, 2) {};
		\node [style=none] (4) at (-4.25, 2) {$,$};
		\node [style=none] (5) at (-1.5, 2) {$R_{2}^{isg}=$};
		\node [style=none] (6) at (-0.75, 2) {};
		\node [style=none] (7) at (-0.25, 2) {};
		\node [style=none] (8) at (0.25, 2) {};
		\node [style=none] (10) at (-5, 0.75) {$F^{isg}=$};
		\node [style=none] (17) at (-6, -0.25) {};
		\node [style=none] (18) at (-4, -0.25) {};
		\node [style=none] (19) at (-6, -1) {};
		\node [style=none] (20) at (-4, -1) {};
		\node [style=none] (22) at (-7, -0.5) {$\rho^{isg}(\sigma_1)=$};
		\node [style=none] (23) at (-3.75, -0.5) {$,$};
		\node [style=none] (24) at (-1.5, -0.25) {};
		\node [style=none] (25) at (0.5, -0.25) {};
		\node [style=none] (26) at (-1.5, -1) {};
		\node [style=none] (27) at (0.5, -1) {};
		\node [style=none] (29) at (-2.5, -0.5) {$\rho^{isg}(\sigma_2)=$};
		\node [style=none] (30) at (0.75, -0.5) {$,$};
		\node [style=none] (31) at (-3.75, -1.75) {};
		\node [style=none] (32) at (-1.75, -1.75) {};
		\node [style=none] (33) at (-3.75, -2.5) {};
		\node [style=none] (34) at (-1.75, -2.5) {};
		\node [style=none] (36) at (-4.75, -2) {$\rho^{isg}(\sigma_3)=$};
		\node [style=none] (38) at (-1.25, -3.25) {};
		\node [style=none] (39) at (0.75, -3.25) {};
		\node [style=none] (40) at (-1.25, -4) {};
		\node [style=none] (41) at (0.75, -4) {};
		\node [style=none] (43) at (-2.25, -3.5) {$\rho^{isg}(\sigma_5)=$};
		\node [style=none] (44) at (-6, -3.25) {};
		\node [style=none] (45) at (-4, -3.25) {};
		\node [style=none] (46) at (-6, -4) {};
		\node [style=none] (47) at (-4, -4) {};
		\node [style=none] (49) at (-7, -3.5) {$\rho^{isg}(\sigma_4)=$};
		\node [style=none] (51) at (0.75, 2) {};
		\node [style=none] (52) at (0.75, 2) {$,$};
		\node [style=none] (54) at (-2.25, -1.75) {};
		\node [style=none] (55) at (-2.25, -2.5) {};
		\node [style=none] (56) at (-3.25, -1.75) {};
		\node [style=none] (57) at (-3.25, -2.5) {};
		\node [style=Z phase dot] (58) at (-5.25, 2) {$\frac{-\pi}{2}$};
		\node [style=X phase dot] (59) at (-0.25, 2) {$\frac{\pi}{2}$};
		\node [style=Z phase dot] (60) at (-3.75, 0.75) {$\frac{\pi}{2}$};
		\node [style=Z phase dot] (61) at (-1.5, 0.75) {$\frac{\pi}{2}$};
		\node [style=X phase dot] (62) at (-2.75, 0.75) {$\frac{\pi}{2}$};
		\node [style=X phase dot] (63) at (-0.5, -0.25) {$\frac{\pi}{2}$};
		\node [style=X phase dot] (64) at (-3.25, -1.75) {};
		\node [style=X phase dot] (65) at (-2.25, -1.75) {};
		\node [style=X phase dot] (67) at (-5, -4) {$\frac{\pi}{2}$};
		\node [style=Z phase dot] (68) at (-5, -0.25) {$\frac{-\pi}{2}$};
		\node [style=Z phase dot] (69) at (-3.25, -2.5) {};
		\node [style=Z phase dot] (70) at (-2.25, -2.5) {};
		\node [style=Z phase dot] (71) at (-2.75, -1.75) {$\frac{\pi}{2}$};
		\node [style=Z phase dot] (72) at (-0.25, -4) {$\frac{-\pi}{2}$};
		\node [style=none] (73) at (-3.25, -3.5) {};
		\node [style=none] (74) at (-3.75, -3.5) {$,$};
		\node [style=none] (75) at (-4.25, 0.75) {};
		\node [style=none] (76) at (-0.75, 0.75) {};
	\end{pgfonlayer}
	\begin{pgfonlayer}{edgelayer}
		\draw (17.center) to (18.center);
		\draw (19.center) to (20.center);
		\draw (24.center) to (25.center);
		\draw (26.center) to (27.center);
		\draw (31.center) to (32.center);
		\draw (33.center) to (34.center);
		\draw (44.center) to (45.center);
		\draw (46.center) to (47.center);
		\draw (38.center) to (39.center);
		\draw (40.center) to (41.center);
		\draw (56.center) to (57.center);
		\draw (54.center) to (55.center);
		\draw (1.center) to (3.center);
		\draw (6.center) to (8.center);
		\draw (60) to (61);
		\draw (75.center) to (76.center);
	\end{pgfonlayer}
\end{tikzpicture}
\caption{The ZX-representation of Ising anyons gates.}\label{fig:zx-of-ising}
\end{figure}

For braid groups on 6-strands or 2-qubit gates, we need to find the matrix representation of the braid group's generators, $B_6$. While it might be conceivable that any set of generators on six strands is a union of generators on three strands, to find out the braid generator for the middle two strands, i.e. strands number 3 and 4, we have to work with $B_6$.  Based on the general relation between Temperley-Lieb algebra and braid groups, the authors of \cite{fan2010braid} present a 4-dimensional representation of the Ising braid group on six strands or anyons as follows:
\begin{align*}
	& \rho^{isg}(\sigma_1) = exp(\frac{\pi i}{8}) \begin{pmatrix}
		-1 & 0 & 0 & 0 \\ 0 & -1 & 0 & 0 \\ 0 & 0 & i & 0 \\ 0 & 0 & 0 & i
	\end{pmatrix}, \\
	& \rho^{isg}(\sigma_2) = -\frac{exp(\frac{-\pi i}{8})}{\sqrt{2}} \begin{pmatrix}
		1 & 0 & i & 0 \\ 0 & 1 & 0 & i \\ i & 0 & 1 & 0 \\ 0 & i & 0 & 1
	\end{pmatrix} \\
	& \rho^{isg}(\sigma_3) = exp(\frac{\pi i}{8}) \begin{pmatrix}
		-1 & 0 & 0 & 0 \\ 0 & i & 0 & 0 \\ 0 & 0 & i & 0 \\ 0 & 0 & 0 & -1
	\end{pmatrix}, \\
	& \rho^{isg}(\sigma_4) = -\frac{exp(\frac{-\pi i}{8})}{\sqrt{2}} \begin{pmatrix}
		1 & i & 0 & 0 \\ i & 1 & 0 & 0 \\ 0 & 0 & 1 & i \\ 0 & 0 & i & 1
	\end{pmatrix} \\
	& \rho^{isg}(\sigma_5) = exp(\frac{\pi i}{8}) \begin{pmatrix}
		-1 & 0 & 0 & 0 \\ 0 & i & 0 & 0 \\ 0 & 0 & -1 & 0 \\ 0 & 0 & 0 & i
	\end{pmatrix} &
\end{align*}
To simplify the block diagonal representations presented above and later to translate them to the ZX forms straightforwardly, we extract two basic generators for constructing these matrices, called $U$-generators.
\begin{align*}
	& U_1^{isg} = exp(\frac{\pi i}{4}) \begin{pmatrix}
		1 & 0 \\ 0 & -i
	\end{pmatrix} = Z_{\frac{-\pi}{2}}, \\
	& U_2^{isg} = \frac{1}{\sqrt{2}} \begin{pmatrix}
		1 & i \\ i & 1 
	\end{pmatrix} = X_{\frac{-\pi}{2}}
\end{align*}
Given U-generators, we restate the 4-dimensional representations of the Ising braid group on 6 anyons as below ($\approx$ means up to a global phase factor),
\begin{align*}
	& \rho^{isg}(\sigma_1) = -exp(\frac{i\pi}{8})(U_1^{isg} \otimes I) \approx ( Z_{\frac{-\pi}{2}} \otimes I ), \\
	& \rho^{isg}(\sigma_2) = -exp(\frac{-i\pi}{8})(U_2^{isg} \otimes I) \approx ( X_{\frac{-\pi}{2}} \otimes I ), \\
	&\rho^{isg}(\sigma_3) \approx  (CNOT)(Z_{-\frac{\pi}{2}}\otimes I)(CNOT), \\
	& \rho^{isg}(\sigma_4) = -exp(\frac{i\pi}{8}) (I \otimes U_2^{isg}) \approx (I \otimes X_{\frac{-\pi}{2}} ) \\
	& \rho^{isg}(\sigma_5) = -exp(\frac{-i\pi}{8}) (I \otimes U_1^{isg}) \approx (I \otimes Z_{\frac{-\pi}{2}} )
\end{align*}


These generators satisfy the following braid relations: 
\begin{align}
& [\rho^{isg}(\sigma_i), \rho^{isg}(\sigma_j)]=0  \hspace{0.5cm}\forall \abs{i-j}\ge 2, \label{first-row}\\
& \rho^{isg}(\sigma_i)\rho^{isg}(\sigma_j)\rho^{isg}(\sigma_i) = \rho^{isg}(\sigma_j)\rho^{isg}(\sigma_i)\rho^{isg}(\sigma_j) & \forall \abs{i-j}= 1,\label{second-row}
\end{align}
The first Equation~\ref{first-row}, states that non-overlapping braids commute. These can be shown straightforwardly using the ZX representation because either: (i) the ZX representations act on different wires or (ii) the representations act on the same wire, but the ZX generators commute thanks to the spider fusion law. For example, the commutation of $\rho^{isg}(\sigma_3)$ and $\rho^{isg}(\sigma_5)$ can be proven as follows:
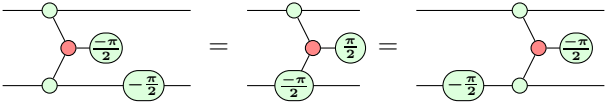
\[
\begin{tikzpicture}[tikzfig]
	\begin{pgfonlayer}{nodelayer}
		\node [style=Z phase dot] (7) at (-2.75, 0) {$\frac{-\pi}{2}$};
		\node [style=X dot] (8) at (-3.75, 0) {};
		\node [style=Z dot] (9) at (-4.25, -1) {};
		\node [style=Z dot] (10) at (-4.25, 1) {};
		\node [style=none] (11) at (-5.5, 1) {};
		\node [style=none] (12) at (-0.5, 1) {};
		\node [style=none] (13) at (-5.5, -1) {};
		\node [style=none] (14) at (-0.5, -1) {};
		\node [style=Z phase dot] (15) at (-1.75, -1) {\,$-\frac{\pi}{2}$\,};
		\node [style=none] (16) at (0.25, 0) {$=$};
		\node [style=Z phase dot] (17) at (3.75, 0) {$\frac{\pi}{2}$};
		\node [style=X dot] (18) at (2.75, 0) {};
		\node [style=Z dot] (20) at (2.25, 1) {};
		\node [style=none] (21) at (1, 1) {};
		\node [style=none] (22) at (4, 1) {};
		\node [style=none] (23) at (1, -1) {};
		\node [style=none] (24) at (4, -1) {};
		\node [style=Z phase dot] (25) at (2.25, -1) {\,$\frac{-\pi}{2}$\,};
		\node [style=none] (26) at (4.75, 0) {$=$};
		\node [style=Z phase dot] (27) at (9.75, 0) {$\frac{-\pi}{2}$};
		\node [style=X dot] (28) at (8.75, 0) {};
		\node [style=Z dot] (29) at (8.25, -1) {};
		\node [style=Z dot] (30) at (8.25, 1) {};
		\node [style=none] (31) at (5.5, 1) {};
		\node [style=none] (32) at (10.5, 1) {};
		\node [style=none] (33) at (5.5, -1) {};
		\node [style=none] (34) at (10.5, -1) {};
		\node [style=Z phase dot] (35) at (6.75, -1) {\,$-\frac{\pi}{2}$\,};
	\end{pgfonlayer}
	\begin{pgfonlayer}{edgelayer}
		\draw (11.center) to (12.center);
		\draw (13.center) to (14.center);
		\draw (9) to (8);
		\draw (10) to (8);
		\draw (8) to (7);
		\draw (21.center) to (22.center);
		\draw (23.center) to (24.center);
		\draw (20) to (18);
		\draw (18) to (17);
		\draw (31.center) to (32.center);
		\draw (33.center) to (34.center);
		\draw (29) to (28);
		\draw (30) to (28);
		\draw (28) to (27);
		\draw (25) to (18);
	\end{pgfonlayer}
\end{tikzpicture}
\]

The second Equation~\ref{second-row}, is essentially the Yang-Baxter equation. These are all implied by variations of the $P$-rule, as shown in Sections~\ref{sec:braid-relations} and~\ref{sec:2q-braids}.

The Ising model is a non-universal model of quantum computation \cite{wang2010, rowell2016}. The above ZX-based relations, furthermore, concede they represent the Clifford fragment of the ZX-calculus, which indicates  non-universality. 
\subsubsection{Fibonacci Model}
Fibonacci anyons present the simplest universal model of topological quantum computation. The label set has only one non-trivial anyon called the Fibonacci anyon, $\{\tau\}$. There is only one non-trivial Fibonacci rule, namely when the Fibonacci anyon fuses with another Fibonacci anyon, $\tau \otimes \tau = 1 \oplus \tau $. 

A technique for efficiently approximating single qubit gates and $CNOT$ was put forth in \cite{bonesteel2005braid}. The universality of this model is also proved in e.g.~\cite{wang2010}.  Further proposals for quantum gate synthesis based on a Monte Carlo algorithm or reinforcement learning were suggested in~\cite{zhang2020}. Here, we take a somewhat dual approach. Rather than translating quantum computational primitives to complex sequences of braid generators, we translate individual braid generators to ZX-diagrams, which is more closely related to quantum gates and can be reasoned about using the rules of the ZX-calculus.

To encode qubits, we have to consider three Fibonacci anyons when their total charge is $\tau$. The internal charge $x$ determines two basis states. We map $\{|x=1\rangle , |x=\tau \rangle \}$ to $\{|0\rangle , |1 \rangle \}$ respectively. Solving pentagonal and hexagonal equations in this Hilbert space, $H = Span(\{|0\rangle, |1\rangle \})$, we obtain the following solutions for $F$ and $R_1^{Fib}$ matrices \cite{wang2010}, if $\phi = \frac{\sqrt{5}+1}{2}$. 
\begin{align}
	& \phi^2 = \phi + 1, 1 = \phi^{-1}+\phi^{-2}, \\
	& F = \begin{pmatrix}
		\phi^{-1} & \phi^{\frac{-1}{2}} \\ \phi^{\frac{-1}{2}} & -\phi^{-1}
	\end{pmatrix}= \begin{pmatrix}
		1 & 0 \\ 0 & i 
	\end{pmatrix} \begin{pmatrix}
		\phi^{-1} & -\phi^{\frac{-1}{2}}i \\ -\phi^{\frac{-1}{2}}i & \phi^{-1} 
	\end{pmatrix}  \begin{pmatrix}
		1 & 0 \\ 0 & i 
	\end{pmatrix}\\
	& R_1^{Fib} = exp(\frac{-4\pi i}{5}) \begin{pmatrix}
		1 & 0 \\ 0 & exp(\frac{7\pi i}{5})
	\end{pmatrix}, R = exp(\frac{7\pi i}{5}) \\
	& R_2^{Fib} = F^{Fib} R_1^{Fib}F^{Fib} = \begin{pmatrix}
		\phi^{-2}+R\phi^{-1} & \phi^\frac{-3}{2}(1-R) \\
		\phi^\frac{-3}{2}(1-R) & \phi^{-1}+R\phi^{-2}
	\end{pmatrix}
\end{align}
One can easily observe $F = Z_{\frac{\pi}{2}} X_\theta Z_{\frac{\pi}{2}}$ and 
$R_1^{Fib}$ is a $\frac{7\pi i}{5}$ green spider. Considering the general form of $X_\theta$-rotation, 
\begin{equation*}
	X_\theta = exp(\frac{\theta i}{2}) \begin{pmatrix}
		cos(\frac{\theta}{2}) & -isin(\frac{\theta}{2}) \\
		-isin(\frac{\theta}{2}) & cos(\frac{\theta}{2})
	\end{pmatrix}
\end{equation*}
we can find $\theta$ for $F^{Fib}$, which is $\theta = 2Arccos(\phi^{-1})$ if $\frac{\pi}{2}\prec \theta \prec \pi$.  Numerically, it is $\theta \approx \frac{129 \pi}{224}$.

Based on the explicit representation of $R_1$ and $R_2$, we can see the following equalities which somehow simplify our next equations.


To find the 4-dimensional representations of Fibonacci anyons, we need to create two pairs of three anyons; the total charge of these six anyons is vacuum or 1. However, the charge of a group of three anyons can be either 1 or $\tau$. If it is $\tau$, we are in the computational basis, and if it is 1, we are in the non-computational basis. As it should be clear from the graph, the computational space is 4-dimensional, and the non-computational space is one-dimensional.

Unlike Ising anyons, Fibonacci anyons have the possibility of leakage into the non-computational space. That is while swapping the middle anyons, anyons number 33 and 44, the state may change to a non-computational state. Note that this is admissible since it preserves the total charge 11 (parent node). So representations of braid group on six anyons should be actually mapped into this 5-dimensional space of $H=H_{C} \oplus H_{NC} $, or $H=Span\{ |NC \rangle, |11 \rangle, |1\tau \rangle, |1\tau \rangle, |\tau1 \rangle, |\tau \tau \rangle \}$ .

The same method of Temperly-Lieb algebra \cite{cui2019search} gives a representation of Fibonacci $B_6$. If $R = exp(\frac{7\pi i}{5})$. 
\begin{align}
	& \rho^{Fib}(\sigma_1) = exp(\frac{-4\pi i}{5}) [R \oplus (R_1^{Fib} \otimes I_2)] \label{eq:fib-1d-1}\\
	& \rho^{Fib}(\sigma_2) = exp(\frac{-4\pi i}{5}) [R \oplus (FR_1^{Fib}F \otimes I_2)] \\
	& \rho^{Fib}(\sigma_3) = exp(\frac{-4\pi i}{5}) [
	P_{14}(R \oplus R_1^{Fib} \oplus FR_1^{Fib}F)P_{14}]\\
	&  \rho^{Fib}(\sigma_4) = exp(\frac{-4\pi i}{5}) [R \oplus (I_2 \otimes FR_1^{Fib}F )] \\
	& \rho^{Fib}(\sigma_5) = exp(\frac{-4\pi i}{5}) [R \oplus (I_2 \otimes R_1^{Fib} )] \label{eq:fib-1d-2}
\end{align}
Unlike the 2-dimensional generators, the 5-dimensional ones do not directly give us clue into their ZX diagrams. However, as one can observe from their matrix form, they are block diagonal and easily re presentable by quantum circuits. In the following part, we intend to find these circuits and then transform them into the ZX diagrams.

We first need to cope with a dimension mismatch. The 6-anyon Fibonacci generators occupy a 5-dimensional space. However, the ZX-calculus, and indeed the usual circuit model, can only represent linear operators on $2^k$-dimensional spaces. Hence, we can further ``pad out'' this 5D space to make it $2^3 = 8$-dimensional, by introducing three extra ``garbage'' basis elements. These have no physical interpretation beyond allowing us to represent the generators above in a subspace of 3-qubit space.

We map the old basis states into new basis states in the following way:
\begin{align*}
	& |NC\rangle \longrightarrow |011\rangle, 
	|11\rangle \longrightarrow |100\rangle, 
	|1\tau\rangle \longrightarrow |101\rangle, \\
	& |\tau1\rangle \longrightarrow |110\rangle, 
	|\tau \tau\rangle \longrightarrow |111\rangle
\end{align*}
Note that this mapping sends basis states in the computation subspace to basis states whose first qubit is $1$, and reserves $|011\rangle$ for the non-computational state, and has extra ``garbage'' basis states $|000\rangle$, $|001\rangle$, and $|010\rangle$.

In our encoding into this larger space, we perform an additional simplification step. The na\"ive encoding of the 5D generators is simply to take the direct sum with the $3\times 3$ identity matrix $I_3$. However, since the garbage basis states will have no significance for our calculation, taking the direct sum with \textit{any} unitary matrix will do just as well. We take advantage of this for the first generator, where we instead take the direct sum with $RI_3$. This yields a much simpler unitary in the 3-qubit representation. The embedded matrices allow us to identify the building blocks or $U$-generators.

\begin{align*}
	& U_1^{Fib}= \begin{pmatrix} RI_4 & 0 \\ 0 & I_4 \end{pmatrix},
	U_2^{Fib} = \begin{pmatrix} I_6 & 0 \\ 0 & RI_2 \end{pmatrix}, \text{where}~R = \frac{7\pi i}{5} \\
	& U_3^{Fib} = \begin{pmatrix} I_2 & 0 \\ 0 & R_2^{Fib} \end{pmatrix} \otimes I_2, 
	U_4^{Fib} = \begin{pmatrix} 
		I_4 & 0 & 0 \\
		0 & R_1^{Fib} & 0 \\
		0 & 0 & I_2 
	\end{pmatrix},\\
	&U_5^{Fib} = \begin{pmatrix}
		I_6 & 0 \\ 0 & R_1^{Fib}
	\end{pmatrix}, 
	U_6 = \begin{pmatrix}
		I_4 & 0 & 0 \\ 0 & R_2^{Fib} & 0 \\ 0 & 0 & I_2
	\end{pmatrix},  \\
	& U_7^{Fib} = \begin{pmatrix}
		I_6 & 0 \\ 0 & R_2^{Fib}
	\end{pmatrix}, \\
	& P_{14} = (CCX_2)(CCX_0)(CCX_2),
	CCX_i: i-\text{target qubit}   
\end{align*}
The ZX-form of $U$ generators for Fibonacci anyons are longer chains of single and 2-qubit gates. To see them, consult PyZX Anyon package \cite{github}.
Fibonacci anyons' braid generators satisfy the following commutation relations, which are also provable by graphical calculus. Consider, $\{1, 2, ..., 7\}$, 
\begin{align*}
	& [U_1, U_i] = 0 && \forall i \in l, & \\
	& [U_2, U_i] = 0 && \forall i \in l/\{3\}, & U_2U_3U_2 = U_3U_2U_3,  &\\
	& [U_3, U_i] = 0 && \forall i,  & \\
	& [U_4, U_i] = 0 && \forall i \in l/\{3, 6\}, & U_4U_6U_4 = U_6U_4U_6, & \\
	& [U_5, U_i] = 0 && \forall i \in l/\{3, 7\}, & U_5U_7U_5 = U_7U_5U_7,  &\\
	& [U_6, U_i] = 0 && \forall i \in l/\{3, 4\}, & \\
\end{align*}

We break down generators of $\mathcal{B}_6$ to the multiplication of $U$-generators as follows; 
\begin{align*}
	& \rho^{Fib}(\sigma_1)=U_1U_2, && \rho^{Fib}(\sigma_2)=U_1U_3, \\
 & \rho^{Fib}(\sigma_3)=P_{14}(U_1U_4U_7)P_{14},\\
	& \rho^{Fib}(\sigma_4)=U_1U_6U_7, && \rho^{Fib}(\sigma_5)=U_1U_4U_5.&
\end{align*}
\begin{figure}
\begin{tikzpicture}
	\begin{pgfonlayer}{nodelayer}
		\node [style=X phase dot] (0) at (0, 0) {$\theta$};
		\node [style=Z phase dot] (1) at (0.5, 0) {$\frac{\pi}{2}$};
		\node [style=Z phase dot] (2) at (-0.5, 0) {$\frac{\pi}{2}$};
		\node [style=none] (3) at (1, 0) {};
		\node [style=none] (4) at (-1, 0) {};
		\node[style=none] (5) at(-1.5, 0){$F=$};
	\end{pgfonlayer}
	\begin{pgfonlayer}{edgelayer}
		\draw (4.center) to (3.center);
	\end{pgfonlayer}
\end{tikzpicture}\\
\begin{tikzpicture}
	\begin{pgfonlayer}{nodelayer}
		\node [style=none] (0) at (-2.75, 0) {$R_1^{Fib}=$};
		\node [style=none] (1) at (-2, 0) {};
		\node [style=none] (2) at (-1.5, 0) {};
		\node [style=none] (3) at (-2, 0) {};
		\node [style=none] (4) at (-1, 0) {};
		\node [style=none] (5) at (-0.5, 0) {};
		\node [style=none] (7) at (1.25, 0) {$R_2^{Fib}$};
		\node [style=none] (8) at (2, 0) {};
		\node [style=none] (9) at (2.5, 0) {};
		\node [style=none] (10) at (3, 0) {};
		\node [style=none] (11) at (3.5, 0) {};
		\node [style=none] (12) at (4, 0) {};
		\node [style=none] (13) at (4.5, 0) {};
		\node [style=none] (14) at (5, 0) {};
		\node [style=Z phase dot] (15) at (-1.5, 0) {$\alpha$};
		\node [style=Z phase dot] (16) at (-1, 0) {$\pi$};
		\node [style=Z phase dot] (17) at (2.5, 0) {$\frac{\pi}{2}$};
		\node [style=Z phase dot] (18) at (3.5, 0) {$\alpha$};
		\node [style=Z phase dot] (19) at (4.5, 0) {$\frac{\pi}{2}$};
		\node [style=X phase dot] (20) at (3, 0) {$\theta$};
		\node [style=X phase dot] (21) at (4, 0) {$\theta$};
		\node [style=none] (22) at (0, 0) {};
		\node [style=none] (23) at (0, 0) {$,$};
	\end{pgfonlayer}
	\begin{pgfonlayer}{edgelayer}
		\draw (8.center) to (14.center);
		\draw (3.center) to (5.center);
	\end{pgfonlayer}
\end{tikzpicture}\\
\begin{tikzpicture}
	\begin{pgfonlayer}{nodelayer}
		\node [style=none] (0) at (-4, 0.5) {};
		\node [style=none] (1) at (-2, 0.5) {};
		\node [style=none] (2) at (-4, 0) {};
		\node [style=none] (3) at (-2, 0) {};
		\node [style=none] (4) at (-4, -0.5) {};
		\node [style=none] (5) at (-2, -0.5) {};
		\node [style=none] (6) at (-4.75, 0) {$U_1=$};
		\node [style=none] (7) at (-1.25, 0) {};
		\node [style=none] (8) at (0.25, 0.5) {};
		\node [style=none] (9) at (2.5, 0.5) {};
		\node [style=none] (10) at (0.25, 0) {};
		\node [style=none] (11) at (2.5, 0) {};
		\node [style=none] (12) at (0.25, -0.5) {};
		\node [style=none] (13) at (2.5, -0.5) {};
		\node [style=none] (14) at (-0.5, 0) {$U_2=$};
		\node [style=Z phase dot] (15) at (-3, 0.5) {$\frac{7\pi}{5}$};
		\node [style=Z phase dot] (17) at (2, 0.25) {\,$\frac{-7\pi}{10}$\,};
		\node [style=Z phase dot] (19) at (1, 0.5) {\,$\frac{7\pi}{5}$\,};
		\node [style=Z phase dot] (20) at (1, 0) {\,$\frac{7\pi}{10}$\,};
		\node [style=X phase dot] (23) at (-3.5, 0.5) {$\pi$};
		\node [style=X phase dot] (24) at (-2.5, 0.5) {$\pi$};
		\node [style=X dot] (25) at (1.5, 0.25) {};
	\end{pgfonlayer}
	\begin{pgfonlayer}{edgelayer}
		\draw (0.center) to (1.center);
		\draw (2.center) to (3.center);
		\draw (4.center) to (5.center);
		\draw (8.center) to (9.center);
		\draw (10.center) to (11.center);
		\draw (12.center) to (13.center);
		\draw (20) to (25);
		\draw (19) to (25);
		\draw (25) to (17);
	\end{pgfonlayer}
\end{tikzpicture}
\caption{The ZX-representation of Fibonacci anyons.}
\end{figure}
It is, therefore, enough to find quantum circuits of the $U$-generators and transform them into the $ZX$-diagrams. One can observe how the expression of $U_6$ and $U_7$ are related to each other by taking $U_7  (CCR_2^{Fib})$ and $U_6 = (I \otimes X \otimes I)U_7(I \otimes X \otimes I)$. We represent $U_7 = (CCF)U_5(CCF)$, that means the ZX-representation of $U_7$ reduces to representing $CCF$ or $(CCZ_{\frac{-\pi}{2}} )(CCX_\theta) (CCZ_{\frac{-\pi}{2}}) $. 

So far, we have explored the full description of braid generators up to six strands. Given these ZX-diagrams, one can create any braid of up to six strands and exploit rewrite rules to simplify them. Using anyon library of PyZX, this process is automatic \cite{github}. However, due to the current limitation of PyZX, after simplification, PyZX does not extract a circuit diagram from a simplified graph since the Fibonacci angle is an arbitrary non-Clifford angle. 

Note that PyZX does an acceptable job when dealing with single qubit gates and braids on three strands. However, we make this simplification more precise by introducing specific $P$-rules for Fibonacci and Ising. The specific $P$-rules not only help to simplify braids, but also it is helpful for topological quantum compiling of single-qubit gates; one can use an intermediate step, switch to the ZX-calculus, optimise them exactly, and come back to the circuit representation. 

\begin{theorem}(\textbf{Fibonacci Single Qubit P-rule})\label{them:fib-p-rule}
For $\theta = 2Arccos(\phi^{-1})$, we have:
\begin{center}
\begin{tikzpicture}[tikzfig]
	\begin{pgfonlayer}{nodelayer}
		\node [style=X phase dot] (0) at (-4, 0) {$\theta$};
		\node [style=Z phase dot] (1) at (-3, 0) {\,$\frac{2\pi}{5}$\,};
		\node [style=X phase dot] (2) at (-2, 0) {$\theta$};
		\node [style=none] (3) at (-5, 0) {};
		\node [style=none] (4) at (-1, 0) {};
		\node [style=none] (5) at (0, 0) {$=$};
		\node [style=X phase dot] (6) at (3, 0) {$\theta$};
		\node [style=Z phase dot] (7) at (2, 0) {\,$\frac{3\pi}{5}$\,};
		\node [style=none] (9) at (1, 0) {};
		\node [style=none] (10) at (5, 0) {};
		\node [style=Z phase dot] (11) at (4, 0) {\,$\frac{3\pi}{5}$\,};
	\end{pgfonlayer}
	\begin{pgfonlayer}{edgelayer}
		\draw (3.center) to (4.center);
		\draw (9.center) to (10.center);
	\end{pgfonlayer}
\end{tikzpicture}
\end{center}
\end{theorem} 
The proof uses the main $P$-rule, the fact that $\phi$ is the golden ratio and satisfies $\phi^2 - \phi - 1= 0$, also $R=\frac{7\pi}{5}$, and that $R$ and $\phi$ are related to each other through the hexagonal equation. As Coecke and Wang mentioned in \cite{coecke2018zx}, if the first and third angles on the left-hand side of the equation are equal, then the first and third angles on the right-hand side must also be equal.  To prove the above rule, we need a lemma. 
\begin{lemma}
If $\alpha= \frac{2\pi}{5}$ and $\phi$ is the golden ratio, then $cos(\alpha)=\frac{\phi^{-1}}{2}$. 
\end{lemma}
\begin{proof}
We define a general $R_m = \begin{pmatrix} R_1 & 0 \\ 0 & R_{\tau}\end{pmatrix}$, and if we substitute it in the hexagonal equation, we have: 
\begin{align*}
& \phi^{-1} + R_\tau = R_1^2, && \phi^{-1}-\phi^{-1}R_\tau = R_1R_\tau, \\ & R^2_\tau+\phi^{-1}R_\tau+1 = 0.&
\end{align*} 
	Considering this set of equations and $cos(\xi)= \frac{-\phi^{-1}}{2}$, we obtain 
$\beta = \frac{\pi}{2} - \frac{\xi}{2}$ and $\frac{R_\tau}{R_1} = \xi - \beta = \frac{3\xi}{2}-\frac{\pi}{2}$. 
\end{proof}
\begin{proof}[Proof of Theorem \ref{them:fib-p-rule}]
Let us substitute initial angles in $(z, z')$, 
	\begin{align*}
		& z = cos(\frac{\alpha}{2}) cos(\theta)+isin(\frac{\alpha}{2}) \\
		& z'= cos(\frac{\alpha}{2})sin(\theta)
	\end{align*}
	Using previous Lemma, we have:
	\begin{align*}
		& cos(\frac{\alpha}{2}) = \frac{\sqrt{\phi^{-1}+2}}{2} = \frac{\phi}{2}, \\
		& sin(\frac{\alpha}{2}) = \frac{\sqrt{2-\phi^{-1}}}{2} = \frac{\phi}{2},\\
		& cos(\frac{\theta}{2}) = \phi^{-1}, \\
  &sin(\frac{\theta}{2})=\phi^{\frac{-1}{2}}, 
		sin(\theta)=2\phi^{\frac{-3}{2}}, \\&cos(\theta)= \sqrt{1-4\phi^{-3}}
	\end{align*}
	We can check $\sin(\frac{\gamma}{2})= \phi^{\frac{-1}{2}}$, which incidently proves$\gamma = \theta \pm \pi$. For side angles. we need to only compute argument of $z$ as clearly $arg(z')=0$. We compute $cos(arg(z))$ instead, 
	\begin{align*}
		& cos(arg(z)) = \frac{\frac{\phi}{2} \sqrt{1-4\phi^{-3}}}{\phi^{-1}} = -cos(\alpha)
	\end{align*}
\end{proof}
\subsection{Single-qubit Braid Relations}\label{sec:braid-relations}
In this part, we can check the non-trivial braid relations on three anyons, namely Yang-Baxter equation. We show that explicitly by using $P$-rule for Ising anyons and Fibonacci. We see that the Ising Yang-Baxter is exactly an instance of the $P$-rule.  In fact, it can be obtained directly from the Hadamard rule, equation~\eqref{eq:had-rule}, by taking the adjoint of both sides.

\begin{itemize}
	\item For \textbf{Ising anyons}, we have 
\begin{equation}
\centering
\begin{tikzpicture}
	\begin{pgfonlayer}{nodelayer}
		\node [style=none] (50) at (-1.25, 0) {$R_1R_2R_1=$};
		\node [style=none] (51) at (-0.25, 0) {};
		\node [style=none] (52) at (0.25, 0) {};
		\node [style=none] (53) at (0.75, 0) {};
		\node [style=none] (54) at (1.25, 0) {};
		\node [style=none] (55) at (1.75, 0) {};
		\node [style=Z phase dot] (56) at (0.25, 0) {$\frac{-\pi}{2}$};
		\node [style=Z phase dot] (57) at (1.25, 0) {$\frac{-\pi}{2}$};
		\node [style=X phase dot] (58) at (0.75, 0) {$\frac{-\pi}{2}$};
	\end{pgfonlayer}
	\begin{pgfonlayer}{edgelayer}
		\draw (51.center) to (55.center);
	\end{pgfonlayer}
\end{tikzpicture}
\end{equation}
The other side, $R_1R_2R_1$ results in an equivalent chain of angles,
\begin{equation}
	\centering
	\begin{tikzpicture}
	\begin{pgfonlayer}{nodelayer}
		\node [style=none] (0) at (-1, 0) {$R_2R_1R_2=$};
		\node [style=none] (1) at (0, 0) {};
		\node [style=none] (2) at (0.5, 0) {};
		\node [style=none] (3) at (1, 0) {};
		\node [style=none] (4) at (1.5, 0) {};
		\node [style=none] (5) at (2, 0) {};
		\node [style=X phase dot] (6) at (0.5, 0) {$\frac{-\pi}{2}$};
		\node [style=X phase dot] (7) at (1.5, 0) {$\frac{-\pi}{2}$};
		\node [style=Z phase dot] (8) at (1, 0) {$\frac{-\pi}{2}$};
	\end{pgfonlayer}
	\begin{pgfonlayer}{edgelayer}
		\draw (1.center) to (5.center);
	\end{pgfonlayer}
\end{tikzpicture}
\end{equation}
This shows Yang-Baxter equation is an example of the P-Rule for Ising anyons.

\item For \textbf{Fibonacci anyons}, from combination of the ZX-equivalence of $R_1R_1$ and $R_2R_2$, we have
\begin{equation}
\begin{tikzpicture}
	\begin{pgfonlayer}{nodelayer}
		\node [style=none] (0) at (-1.3, 0.5) {$R_2R_1R_2=$};
		\node [style=none] (1) at (-2.5, 0) {};
		\node [style=none] (2) at (-2.25, 0) {};
		\node [style=none] (3) at (-1.75, 0) {};
		\node [style=none] (4) at (-1.25, 0) {};
		\node [style=none] (5) at (-0.75, 0) {};
		\node [style=none] (6) at (-0.25, 0) {};
		\node [style=none] (7) at (0.25, 0) {};
		\node [style=none] (8) at (0.75, 0) {};
		\node [style=none] (9) at (1.25, 0) {};
		\node [style=none] (10) at (1.75, 0) {};
		\node [style=none] (11) at (2.25, 0) {};
		\node [style=none] (12) at (2.75, 0) {};
		\node [style=none] (13) at (3.25, 0) {};
		\node [style=none] (14) at (3.75, 0) {};
		\node [style=none] (15) at (4.25, 0) {};
		\node [style=none] (16) at (4.75, 0) {};
		\node [style=none] (17) at (5.1, 0) {};
		\node [style=Z phase dot] (18) at (-2.25, 0) {$\frac{\pi}{2}$};
		\node [style=Z phase dot] (19) at (-1.25, 0) {$\frac{\pi}{2}$};
		\node [style=Z phase dot] (20) at (-0.75, 0) {$\frac{7\pi}{5}$};
		\node [style=Z phase dot] (21) at (-0.25, 0) {$\frac{-\pi}{2}$};
		\node [style=Z phase dot] (22) at (0.75, 0) {$\frac{-\pi}{2}$};
		\node [style=Z phase dot] (23) at (1.25, 0) {$\frac{7\pi}{5}$};
		\node [style=Z phase dot] (24) at (1.75, 0) {$\frac{\pi}{2}$};
		\node [style=Z phase dot] (25) at (2.75, 0) {$\frac{\pi}{2}$};
		\node [style=Z phase dot] (26) at (3.25, 0) {$\frac{7\pi}{5}$};
		\node [style=Z phase dot] (27) at (3.75, 0) {$\frac{-\pi}{2}$};
		\node [style=Z phase dot] (28) at (4.75, 0) {$\frac{-\pi}{2}$};
		\node [style=X phase dot] (29) at (-1.75, 0) {$\theta$};
		\node [style=X phase dot] (30) at (0.25, 0) {$-\theta$};
		\node [style=X phase dot] (31) at (2.25, 0) {$\theta$};
		\node [style=X phase dot] (32) at (4.25, 0) {$-\theta$};
	\end{pgfonlayer}
	\begin{pgfonlayer}{edgelayer}
		\draw (1.center) to (17.center);
	\end{pgfonlayer}
\end{tikzpicture}
\end{equation}
We fuse adjacent phases, 
\begin{equation}
	\begin{tikzpicture}
	\begin{pgfonlayer}{nodelayer}
		\node [style=none] (0) at (-2, 0) {$R_2R_1R_2=$};
		\node [style=none] (1) at (-1, 0) {};
		\node [style=none] (2) at (-0.5, 0) {};
		\node [style=none] (3) at (0, 0) {};
		\node [style=none] (4) at (0.5, 0) {};
		\node [style=none] (5) at (1, 0) {};
		\node [style=none] (6) at (1.5, 0) {};
		\node [style=none] (7) at (2, 0) {};
		\node [style=none] (8) at (2.5, 0) {};
		\node [style=none] (9) at (3, 0) {};
		\node [style=none] (10) at (3.5, 0) {};
		\node [style=Z phase dot] (11) at (-0.5, 0) {$\frac{\pi}{2}$};
		\node [style=Z phase dot] (12) at (0.5, 0) {$R$};
		\node [style=Z phase dot] (13) at (1.5, 0) {$R$};
		\node [style=Z phase dot] (14) at (2.5, 0) {$R$};
		\node [style=Z phase dot] (15) at (3.5, 0) {$\frac{-\pi}{2}$};
		\node [style=none] (16) at (4, 0) {};
		\node [style=X phase dot] (17) at (3, 0) {$-\theta$};
		\node [style=X phase dot] (18) at (2, 0) {$\theta$};
		\node [style=X phase dot] (19) at (1, 0) {$-\theta$};
		\node [style=X phase dot] (20) at (0, 0) {$\theta$};
		\node [style=none] (21) at (-0.25, 0.25) {};
		\node [style=none] (22) at (1.25, -0.25) {};
		\node [style=none] (23) at (-0.25, -0.25) {};
		\node [style=none] (24) at (1.25, 0.25) {};
	\end{pgfonlayer}
	\begin{pgfonlayer}{edgelayer}
		\draw (1.center) to (10.center);
		\draw (15) to (16.center);
		\draw [style=hadamard edge] (21.center) to (24.center);
		\draw [style=hadamard edge] (24.center) to (22.center);
		\draw [style=hadamard edge] (22.center) to (23.center);
		\draw [style=hadamard edge] (23.center) to (21.center);
	\end{pgfonlayer}
\end{tikzpicture}
\end{equation}
Applying the Fibonacci P-rule to the boxed angles, we have the following equation, where $R=\alpha+\pi$,
\begin{equation}
	\begin{tikzpicture}[tikzfig]
	\begin{pgfonlayer}{nodelayer}
		\node [style=none] (113) at (-3.75, 0) {$R_2R_1R_2 =  $};
		\node [style=none] (114) at (-1.75, 0) {};
		\node [style=none] (115) at (-1, 0) {};
		\node [style=Z phase dot] (124) at (-1, 0) {$\frac{\pi}{2}$};
		\node [style=Z phase dot] (125) at (1.5, 0) {$\alpha$};
		\node [style=Z phase dot] (126) at (3.5, 0) {$\alpha$};
		\node [style=Z phase dot] (127) at (5.5, 0) {$\alpha$};
		\node [style=Z phase dot] (128) at (7.5, 0) {$\frac{\pi}{2}$};
		\node [style=none] (129) at (9, 0) {};
		\node [style=X phase dot] (130) at (0.75, 0) {$\theta$};
		\node [style=X phase dot] (131) at (2.25, 0) {$\theta$};
		\node [style=X phase dot] (132) at (4.5, 0) {$\theta$};
		\node [style=X phase dot] (133) at (6.25, 0) {$\theta$};
		\node [style=none] (134) at (3, 0.5) {};
		\node [style=none] (135) at (3, -0.5) {};
		\node [style=none] (136) at (0, -0.5) {};
		\node [style=none] (137) at (0, 0.5) {};
		\node [style=none] (138) at (4, 0.5) {};
		\node [style=none] (139) at (7, 0.5) {};
		\node [style=none] (140) at (7, -0.5) {};
		\node [style=none] (141) at (4, -0.5) {};
	\end{pgfonlayer}
	\begin{pgfonlayer}{edgelayer}
		\draw [style=box edge] (137.center) to (134.center);
		\draw [style=box edge] (134.center) to (135.center);
		\draw [style=box edge] (135.center) to (136.center);
		\draw [style=box edge] (136.center) to (137.center);
		\draw [style=box edge] (138.center) to (139.center);
		\draw [style=box edge] (139.center) to (140.center);
		\draw [style=box edge] (139.center) to (140.center);
		\draw [style=box edge] (138.center) to (141.center);
		\draw [style=box edge] (141.center) to (140.center);
		\draw [style=box edge] (140.center) to (139.center);
		\draw [style=box edge] (139.center) to (140.center);
		\draw [style=box edge] (139.center) to (140.center);
		\draw [style=box edge] (139.center) to (140.center);
		\draw (114.center) to (129.center);
	\end{pgfonlayer}
\end{tikzpicture}
\end{equation}
We again apply the $P$-rule to obtain, 
\begin{equation}
\begin{tikzpicture}[tikzfig]
	\begin{pgfonlayer}{nodelayer}
		\node [style=none] (0) at (-4, 1) {$R_2R_1R_2=$};
		\node [style=none] (1) at (-4.5, 0) {};
		\node [style=none] (11) at (10, 0) {};
		\node [style=Z phase dot] (12) at (-3, 0) {$\frac{\pi}{2}$};
		\node [style=Z phase dot] (13) at (-1.25, 0) {$\pi-\alpha$};
		\node [style=Z phase dot] (14) at (1.25, 0) {$\pi-\alpha$};
		\node [style=Z phase dot] (15) at (2.5, 0) {$\alpha$};
		\node [style=Z phase dot] (16) at (3.75, 0) {$\pi-\alpha$};
		\node [style=Z phase dot] (17) at (6.25, 0) {$\pi-\alpha$};
		\node [style=Z phase dot] (18) at (8, 0) {$\frac{\pi}{2}$};
		\node [style=X phase dot] (19) at (5, 0) {$\theta$};
		\node [style=X phase dot] (20) at (0, 0) {$\theta$};
	\end{pgfonlayer}
	\begin{pgfonlayer}{edgelayer}
		\draw (1.center) to (11.center);
	\end{pgfonlayer}
\end{tikzpicture}
\end{equation}
We fuse the middle angles, 
\begin{equation}
	\begin{tikzpicture}[tikzfig]
	\begin{pgfonlayer}{nodelayer}
		\node [style=none] (53) at (-3, 0) {$R_2R_1R_2=$};
		\node [style=none] (54) at (-1.25, 0) {};
		\node [style=none] (55) at (-0.75, 0) {};
		\node [style=none] (56) at (-0.25, 0) {};
		\node [style=none] (57) at (0.25, 0) {};
		\node [style=none] (58) at (3.75, 0) {};
		\node [style=Z phase dot] (91) at (0, 0) {$\frac{\pi}{2}$};
		\node [style=Z phase dot] (92) at (2.5, 0) {$\frac{\pi}{2}$};
		\node [style=X phase dot] (107) at (1.25, 0) {$\theta$};
	\end{pgfonlayer}
	\begin{pgfonlayer}{edgelayer}
		\draw (54.center) to (58.center);
	\end{pgfonlayer}
\end{tikzpicture}
\end{equation}
One can explicitly write the left hand side, $R_1R_2R_1$, and obtain the same equality. 
\end{itemize}
\begin{rem}
Observe that the above results are completely exact. We did not use $\phi$ explicitly. However, anyon library of $PyZX$ works explicitly with angles and gives an approximate for braids on three strands or a composition of Fibonacci single qubit gates. 
\end{rem}
As mentioned before, having a $P$-rule for Fibonacci anyons, one can create a long chain of braids on three strands, and by consecutive application of $P$-rule for Fibonacci, plus other $ZX$-calculus rules, we are able to find a simple equivalent braid or circuit. The following braid is built on the Braid word 
$$B = [R_1,R_1,R_2,R_2,R_2,R_2,R_1,R_1]$$ 
The braid is drawn in the Figure below.
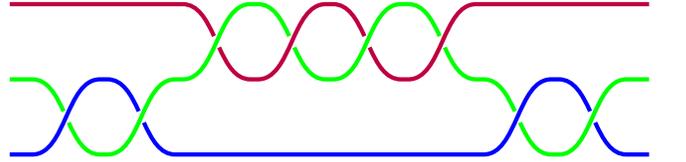
\begin{center}
\begin{tikzpicture}
\pic[
rotate=90,
braid/.cd,
every strand/.style={ultra thick},
strand 1/.style={blue},
strand 2/.style={green},
strand 3/.style={purple},
] {braid={s_1 S_1 S_2 S_2 s_2 s_2 s_1 s_1}};
\end{tikzpicture}
\end{center}	
The same braid in the $ZX$-representation is as follows. In general, to find out
the outcome matrix for this braid, one needs to multiply braid matrices. But
we consider coloured lines as Fibonacci anyons, we are able to simplify the braid
graphically exactly as Figure~\ref{fig:fib-braid-three}. We use a combination of fusion and $\pi$-rules of the ZX, meaning, we fuse spiders with similar colours and if a spider $\pi$ reaches another spider, it changes the sign of the spider and passes. 
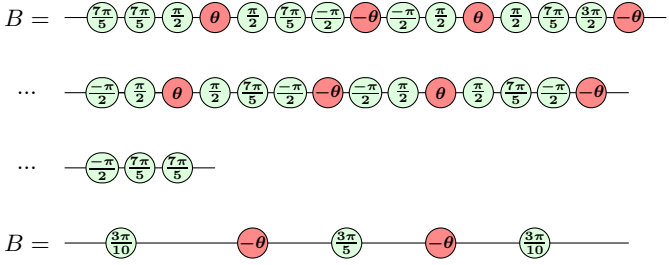
\begin{figure}[!h]
\begin{tikzpicture}
	\begin{pgfonlayer}{nodelayer}
		\node [style=none] (0) at (-3, 2) {$B=$};
		\node [style=Z phase dot] (1) at (-2, 2) {$\frac{7\pi}{5}$};
		\node [style=Z phase dot] (2) at (-1.5, 2) {$\frac{7\pi}{5}$};
		\node [style=Z phase dot] (3) at (-1, 2) {$\frac{\pi}{2}$};
		\node [style=X phase dot] (4) at (-0.5, 2) {$\theta$};
		\node [style=Z phase dot] (5) at (0, 2) {$\frac{\pi}{2}$};
		\node [style=Z phase dot] (6) at (0.5, 2) {$\frac{7\pi}{5}$};
		\node [style=Z phase dot] (7) at (1, 2) {$\frac{-\pi}{2}$};
		\node [style=X phase dot] (8) at (1.5, 2) {$-\theta$};
		\node [style=Z phase dot] (9) at (2, 2) {$\frac{-\pi}{2}$};
		\node [style=Z phase dot] (10) at (2.5, 2) {$\frac{\pi}{2}$};
		\node [style=X phase dot] (11) at (3, 2) {$\theta$};
		\node [style=Z phase dot] (12) at (3.5, 2) {$\frac{\pi}{2}$};
		\node [style=Z phase dot] (13) at (4, 2) {$\frac{7\pi}{5}$};
		\node [style=Z phase dot] (14) at (4.5, 2) {$\frac{3\pi}{2}$};
		\node [style=X phase dot] (15) at (5, 2) {$-\theta$};
		\node [style=none] (16) at (5.5, 2) {};
		\node [style=none] (17) at (6, 2) {$...$};
		\node [style=none] (18) at (-2.5, 2) {};
		\node [style=none] (19) at (-2.5, 1) {};
		\node [style=Z phase dot] (20) at (-2, 1) {$\frac{-\pi}{2}$};
		\node [style=Z phase dot] (21) at (-1.5, 1) {$\frac{\pi}{2}$};
		\node [style=X phase dot] (22) at (-1, 1) {$\theta$};
		\node [style=Z phase dot] (23) at (-0.5, 1) {$\frac{\pi}{2}$};
		\node [style=Z phase dot] (24) at (0, 1) {$\frac{7\pi}{5}$};
		\node [style=Z phase dot] (25) at (0.5, 1) {$\frac{-\pi}{2}$};
		\node [style=X phase dot] (26) at (1, 1) {$-\theta$};
		\node [style=Z phase dot] (27) at (1.5, 1) {$\frac{-\pi}{2}$};
		\node [style=Z phase dot] (28) at (2, 1) {$\frac{\pi}{2}$};
		\node [style=X phase dot] (29) at (2.5, 1) {$\theta$};
		\node [style=Z phase dot] (30) at (3, 1) {$\frac{\pi}{2}$};
		\node [style=Z phase dot] (31) at (3.5, 1) {$\frac{7\pi}{5}$};
           \node [style=none] (47) at (-3, 0) {$...$};
           \node [style=none] (48) at (-2.5, 0) {};
           \node [style=none] (49) at (-0.5, 0) {};
		\node [style=Z phase dot] (32) at (4, 1) {$\frac{-\pi}{2}$};
		\node [style=X phase dot] (33) at (4.5, 1) {$-\theta$};
		\node [style=Z phase dot] (34) at (-2, 0) {$\frac{-\pi}{2}$};
		\node [style=Z phase dot] (35) at (-1.5, 0) {$\frac{7\pi}{5}$};
		\node [style=Z phase dot] (36) at (-1, 0) {$\frac{7\pi}{5}$};
		\node [style=none] (37) at (5, 1) {};
		\node [style=none] (38) at (-3, -1) {$B=$};
		\node [style=none] (39) at (-2.5, -1) {};
		\node [style=Z phase dot] (40) at (-1.75, -1) {$\frac{3\pi}{10}$};
		\node [style=X phase dot] (41) at (0, -1) {$-\theta$};
		\node [style=Z phase dot] (42) at (1.25, -1) {$\frac{3\pi}{5}$};
		\node [style=X phase dot] (43) at (2.5, -1) {$-\theta$};
		\node [style=Z phase dot] (44) at (3.75, -1) {$\frac{3\pi}{10}$};
		\node [style=none] (45) at (5, -1) {};
		\node [style=none] (46) at (-3, 1) {$...$};
	\end{pgfonlayer}
	\begin{pgfonlayer}{edgelayer}
		\draw (44) to (40);
		\draw (39.center) to (40);
		\draw (44) to (45.center);
		\draw (20) to (37.center);
		\draw (19.center) to (22);
		\draw (18.center) to (16.center);
            \draw (48.center) to (34.center);
            \draw (34.center) to (35.center);
            \draw (35.center) to (36.center);
            \draw (36.center) to (49.center);
	\end{pgfonlayer}
\end{tikzpicture}
\caption{Braid on three strands with Fibonacci anyons.}\label{fig:fib-braid-three}
\end{figure}
\subsection{Two-qubit Braid Relations}\label{sec:2q-braids}
Finally, we note that the Yang-Baxter equations for the Ising representation of $B_6$ can also be shown straightforwardly. They are either instances of the $P$-rule applied on just one of the two qubit-wires, or a combination of the $P$-rule and the so-called ``phase gadget'' rules for two qubits. Namely, we have the following equations for any angle $\alpha$:
\begin{figure}[!h]
\begin{tikzpicture}[tikzfig]
	\begin{pgfonlayer}{nodelayer}
		\node [style=none] (8) at (-1.25, 1) {};
		\node [style=none] (9) at (1, 1) {};
		\node [style=none] (10) at (-1.25, -1) {};
		\node [style=none] (11) at (1, -1) {};
		\node [style=Z phase dot] (17) at (1.25, 0) {$\alpha$};
		\node [style=X dot] (25) at (0, 0) {};
		\node [style=Z dot] (26) at (-0.5, -1) {};
		\node [style=Z dot] (27) at (-0.5, 1) {};
		\node [style=none] (28) at (2.25, 0) {$=$};
		\node [style=none] (29) at (3, 1) {};
		\node [style=none] (30) at (6.5, 1) {};
		\node [style=none] (31) at (3, -1) {};
		\node [style=none] (32) at (6.5, -1) {};
		\node [style=Z dot] (33) at (3.75, 1) {};
		\node [style=Z dot] (34) at (5.75, 1) {};
		\node [style=X dot] (35) at (3.75, -1) {};
		\node [style=X dot] (36) at (5.75, -1) {};
		\node [style=Z phase dot] (37) at (4.75, -1) {$\alpha$};
		\node [style=none] (38) at (7.25, 0) {$=$};
		\node [style=none] (39) at (8, -1) {};
		\node [style=none] (40) at (11.5, -1) {};
		\node [style=none] (41) at (8, 1) {};
		\node [style=none] (42) at (11.5, 1) {};
		\node [style=Z dot] (43) at (8.75, -1) {};
		\node [style=Z dot] (44) at (10.75, -1) {};
		\node [style=X dot] (45) at (8.75, 1) {};
		\node [style=X dot] (46) at (10.75, 1) {};
		\node [style=Z phase dot] (47) at (9.75, 1) {$\alpha$};
	\end{pgfonlayer}
	\begin{pgfonlayer}{edgelayer}
		\draw (8.center) to (9.center);
		\draw (10.center) to (11.center);
		\draw (25) to (17);
		\draw (26) to (25);
		\draw (25) to (27);
		\draw (29.center) to (30.center);
		\draw (31.center) to (32.center);
		\draw (33) to (35);
		\draw (34) to (36);
		\draw (39.center) to (40.center);
		\draw (41.center) to (42.center);
		\draw (43) to (45);
		\draw (44) to (46);
	\end{pgfonlayer}
\end{tikzpicture}
\end{figure}
which follow from the other ZX calculus rules (see e.g.~\cite{zxworks}). Now, we can show the Yang-Baxter equation for $\rho^{isg}(\sigma_2)$ and $\rho^{isg}(\sigma_3)$ as shown in Figure~\ref{fig: zx-yang-baxter-ising}. Here, we used the spider-fusion law, the P-rule in the step marked $(*)$, and the fact that two CNOT gates cancel.
\begin{figure}[!h]
\begin{tikzpicture}[tikzfig]
	\begin{pgfonlayer}{nodelayer}
		\node [style=none] (0) at (0.5, 1) {};
		\node [style=none] (1) at (4, 1) {};
		\node [style=none] (2) at (0.5, -1) {};
		\node [style=none] (3) at (4, -1) {};
		\node [style=Z phase dot] (4) at (3.75, 0) {$\frac{-\pi}{2}$};
		\node [style=X dot] (5) at (2.75, 0) {};
		\node [style=Z dot] (6) at (2.25, -1) {};
		\node [style=Z dot] (7) at (2.25, 1) {};
		\node [style=none] (8) at (5, 0) {$=$};
		\node [style=X phase dot] (9) at (1.25, -1) {$\frac{-\pi}{2}$};
		\node [style=none] (10) at (6, 1) {};
		\node [style=none] (11) at (10.5, 1) {};
		\node [style=none] (12) at (6, -1) {};
		\node [style=none] (13) at (10.5, -1) {};
		\node [style=Z phase dot] (14) at (8.25, -1) {$\frac{-\pi}{2}$};
		\node [style=X dot] (15) at (7.5, -1) {};
		\node [style=Z dot] (16) at (7.5, 1) {};
		\node [style=X phase dot] (17) at (6.75, -1) {$\frac{-\pi}{2}$};
		\node [style=X dot] (18) at (9, -1) {};
		\node [style=Z dot] (19) at (9, 1) {};
		\node [style=X phase dot] (20) at (3.25, -1) {$\frac{-\pi}{2}$};
		\node [style=X phase dot] (21) at (9.75, -1) {$\frac{-\pi}{2}$};
		\node [style=none] (22) at (11.25, 0) {$=$};
		\node [style=none] (23) at (12, 1) {};
		\node [style=none] (24) at (17, 1) {};
		\node [style=none] (25) at (12, -1) {};
		\node [style=none] (26) at (17, -1) {};
		\node [style=Z phase dot] (27) at (14.5, -1) {$\frac{-\pi}{2}$};
		\node [style=X dot] (28) at (12.75, -1) {};
		\node [style=Z dot] (29) at (12.75, 1) {};
		\node [style=X phase dot] (30) at (13.5, -1) {$\frac{-\pi}{2}$};
		\node [style=X dot] (31) at (16.25, -1) {};
		\node [style=Z dot] (32) at (16.25, 1) {};
		\node [style=X phase dot] (33) at (15.5, -1) {$\frac{-\pi}{2}$};
		\node [style=none] (34) at (3.25, -3.5) {$=$};
		\node [style=none] (35) at (4, -2.5) {};
		\node [style=none] (36) at (9, -2.5) {};
		\node [style=none] (37) at (4, -4.5) {};
		\node [style=none] (38) at (9, -4.5) {};
		\node [style=X phase dot] (39) at (6.5, -4.5) {$\frac{-\pi}{2}$};
		\node [style=X dot] (40) at (4.75, -4.5) {};
		\node [style=Z dot] (41) at (4.75, -2.5) {};
		\node [style=Z phase dot] (42) at (5.5, -4.5) {$\frac{-\pi}{2}$};
		\node [style=X dot] (43) at (8.25, -4.5) {};
		\node [style=Z dot] (44) at (8.25, -2.5) {};
		\node [style=Z phase dot] (45) at (7.5, -4.5) {$\frac{-\pi}{2}$};
		\node [style=none] (46) at (3.25, -2.75) {$(*)$};
		\node [style=none] (47) at (10, -3.5) {$=$};
		\node [style=none] (48) at (10.75, -2.5) {};
		\node [style=none] (49) at (17, -2.5) {};
		\node [style=none] (50) at (10.75, -4.5) {};
		\node [style=none] (51) at (17, -4.5) {};
		\node [style=X phase dot] (52) at (14.5, -4.5) {$\frac{-\pi}{2}$};
		\node [style=X dot] (53) at (11.5, -4.5) {};
		\node [style=Z dot] (54) at (11.5, -2.5) {};
		\node [style=Z phase dot] (55) at (12.25, -4.5) {$\frac{-\pi}{2}$};
		\node [style=X dot] (56) at (16.25, -4.5) {};
		\node [style=Z dot] (57) at (16.25, -2.5) {};
		\node [style=Z phase dot] (58) at (15.5, -4.5) {$\frac{-\pi}{2}$};
		\node [style=X dot] (59) at (13, -4.5) {};
		\node [style=Z dot] (60) at (13, -2.5) {};
		\node [style=X dot] (61) at (13.75, -4.5) {};
		\node [style=Z dot] (62) at (13.75, -2.5) {};
		\node [style=none] (63) at (3.25, -7) {$=$};
		\node [style=none] (64) at (4, -6) {};
		\node [style=none] (65) at (10, -6) {};
		\node [style=none] (66) at (4, -8) {};
		\node [style=none] (67) at (10, -8) {};
		\node [style=X phase dot] (68) at (7, -8) {$\frac{-\pi}{2}$};
		\node [style=X dot] (69) at (4.75, -8) {};
		\node [style=Z dot] (70) at (4.75, -6) {};
		\node [style=Z phase dot] (71) at (5.5, -8) {$\frac{-\pi}{2}$};
		\node [style=X dot] (72) at (9.25, -8) {};
		\node [style=Z dot] (73) at (9.25, -6) {};
		\node [style=Z phase dot] (74) at (8.5, -8) {$\frac{-\pi}{2}$};
		\node [style=X dot] (75) at (6.25, -8) {};
		\node [style=Z dot] (76) at (6.25, -6) {};
		\node [style=X dot] (77) at (7.75, -8) {};
		\node [style=Z dot] (78) at (7.75, -6) {};
		\node [style=none] (79) at (11.5, -6) {};
		\node [style=none] (80) at (16.5, -6) {};
		\node [style=none] (81) at (11.5, -8) {};
		\node [style=none] (82) at (16.5, -8) {};
		\node [style=Z phase dot] (83) at (17, -7) {$\frac{-\pi}{2}$};
		\node [style=X dot] (84) at (16, -7) {};
		\node [style=Z dot] (85) at (15.5, -8) {};
		\node [style=Z dot] (86) at (15.5, -6) {};
		\node [style=X phase dot] (87) at (14, -8) {$\frac{-\pi}{2}$};
		\node [style=none] (88) at (10.75, -7) {$=$};
		\node [style=Z phase dot] (89) at (13.75, -7) {$\frac{-\pi}{2}$};
		\node [style=X dot] (90) at (12.75, -7) {};
		\node [style=Z dot] (91) at (12.25, -8) {};
		\node [style=Z dot] (92) at (12.25, -6) {};
	\end{pgfonlayer}
	\begin{pgfonlayer}{edgelayer}
		\draw (0.center) to (1.center);
		\draw (2.center) to (3.center);
		\draw (5) to (4);
		\draw (6) to (5);
		\draw (5) to (7);
		\draw (10.center) to (11.center);
		\draw (12.center) to (13.center);
		\draw (15) to (16);
		\draw (18) to (19);
		\draw (23.center) to (24.center);
		\draw (25.center) to (26.center);
		\draw (28) to (29);
		\draw (31) to (32);
		\draw (35.center) to (36.center);
		\draw (37.center) to (38.center);
		\draw (40) to (41);
		\draw (43) to (44);
		\draw (48.center) to (49.center);
		\draw (50.center) to (51.center);
		\draw (53) to (54);
		\draw (56) to (57);
		\draw (59) to (60);
		\draw (61) to (62);
		\draw (64.center) to (65.center);
		\draw (66.center) to (67.center);
		\draw (69) to (70);
		\draw (72) to (73);
		\draw (75) to (76);
		\draw (77) to (78);
		\draw (79.center) to (80.center);
		\draw (81.center) to (82.center);
		\draw (84) to (83);
		\draw (85) to (84);
		\draw (84) to (86);
		\draw (90) to (89);
		\draw (91) to (90);
		\draw (90) to (92);
	\end{pgfonlayer}
\end{tikzpicture}
\caption{The ZX-representation of the Yang-Baxter equation.} \label{fig: zx-yang-baxter-ising}
\end{figure}
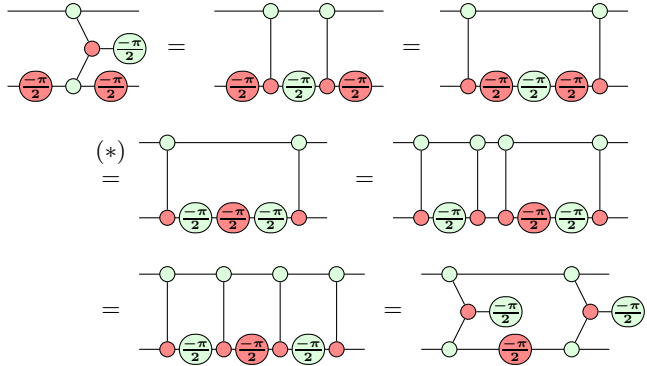
The Yang-Baxter equation for $\rho^{isg}(\sigma_3)$ and $\rho^{isg}(\sigma_4)$ can be proved using the mirror-image of the derivation above.

We leave the full ``ZX-fraction'' of the two-qubit Fibonacci model and graphical proofs of the associated Yang-Baxter equations as future work.

\section{Conclusion and Future Work}
We have demonstrated that the category describing topological quantum computation is actually a subcategory of \hilb. We represented elements of the Fibonacci and Ising models with the ZX-calculus, and showed a new Euler decomposition rule (P-rule) for the single qubit Fibonacci case.

Thanks to the universality of ZX-diagrams, it follows that we can find a graphical representation for any linear map over qubits as a ZX-diagram. This yields very simple representations and proofs in the case of the Ising representation of $B_3$ and $B_6$, encoding 1 and 2 logical qubits, respectively. However, the Fibonacci representation of $B_6$ on three-qubit wires yields unwieldy representations for some of the braid operators as ZX-diagrams, partly due to the need to translate quantum CCZ gates. It could be the case that by switching to a graphical calculus like the ZH-calculus~\cite{backens2018zhcalculus}, which can more elegantly capture CCZ and related constructions, we can more easily represent and work with this representation.

Another interesting area of research is to identify a new \textit{Fibonacci fragment} of the ZX-calculus, consisting of Z-phases that are linear combinations of $\pi/10$ and $\theta = 2Arccos(\phi^{-1})$. The contains the Clifford fragment as well as a non-Clifford phase gate, hence must be universal. It also contains (at least) one new exact P-rule given by Theorem~\ref{them:fib-p-rule}. It would therefore be interesting to see what (if any) other new rules are needed to produce a complete graphical calculus, and whether this suffices for proving any equation involving the ZX representation Fibonacci anyons.

\bibliographystyle{apsrev4-1}
\bibliography{aipsamp.bib}

\end{document}